\def\final{1}
\def\stocsubmit{0}

\documentclass[11pt]{article}
\usepackage{dsfont}
\usepackage{algorithm, algorithmic}
\usepackage{amsmath, amssymb, amsthm}
\usepackage{framed}
\usepackage{fullpage}
\usepackage{graphicx}
\usepackage{verbatim}
\usepackage{float}
\usepackage{color}
\definecolor{DarkGreen}{rgb}{0.1,0.5,0.1}
\definecolor{DarkRed}{rgb}{0.5,0.1,0.1}
\definecolor{DarkBlue}{rgb}{0.1,0.1,0.5}
\usepackage{enumitem}
\usepackage[pdftex]{hyperref}
\hypersetup{
    unicode=false,				
    pdftoolbar=true,				
    pdfmenubar=true,			
    pdffitwindow=false, 		
    pdftitle={},    					
    pdfauthor={}
    pdfsubject={},				
    pdfnewwindow=true,		
    pdfkeywords={},			
    colorlinks=true,				
    linkcolor=DarkRed,		
    citecolor=DarkGreen,	
    filecolor=DarkRed,		
    urlcolor=DarkBlue,		
}

\ifnum\final=0
\newcommand{\mynote}[1]{\marginpar{\tiny \sf #1}}
\else
\newcommand{\mynote}[1]{}
\fi
\newcommand{\mnote}[1]{\mynote{Mark: {#1}}}
\newcommand{\knote}[1]{\mynote{Kobbi: {#1}}}

\newcommand{\remove}[1]{}

\ifnum\stocsubmit=0
\newcommand{\stocrm}[1]{#1}
\newcommand{\stoctext}[1]{}
\else
\newcommand{\stocrm}[1]{}
\newcommand{\stoctext}[1]{#1}
\fi

\newcommand{\LDim}{{\sf Ldim}}
\newcommand{\Lap}{\operatorname{\rm Lap}}
\newcommand{\AAA}{\mathcal A}
\newcommand{\BBB}{\mathcal B}
\newcommand{\FFF}{\mathcal F}
\newcommand{\XXX}{\mathcal X}
\def\opt{\mathop{\rm{opt}}\nolimits}

\newcommand\N{\mathbb{N}}
\newcommand\R{\mathbb{R}}

\newcommand{\cA}{\mathcal{A}}

\newcommand{\cC}{\mathcal{C}}
\newcommand{\cD}{\mathcal{D}}

\newcommand{\cF}{\mathcal{F}}

\newcommand{\cM}{\mathcal{M}}

\newcommand{\cR}{\mathcal{R}}

\newcommand{\tO}{\tilde{O}}

\newcommand{\eps}{\varepsilon}
\newcommand{\poly}{\mathrm{poly}}
\newcommand{\polylog}{\mathrm{polylog}}

\newcommand{\getsr}{\gets_{\mbox{\tiny R}}}

\newtheorem{theorem}{Theorem}[section]
\newtheorem{lemma}[theorem]{Lemma}

\newtheorem{claim}[theorem]{Claim}
\newtheorem{remark}[theorem]{Remark}

\newtheorem{proposition}[theorem]{Proposition}
\newtheorem{observation}[theorem]{Observation}

\theoremstyle{definition}
\newtheorem{definition}[theorem]{Definition}

\newcommand{\db}{D}

\newcommand{\error}{{\rm error}}

\newcommand{\VC}{\operatorname{\rm VC}}

\newcommand{\RepDim}{\operatorname{\rm RepDim}}

\newcommand{\point}{\operatorname{\tt POINT}}

\newcommand{\thresh}{\operatorname*{\tt THRESH}}

\newcommand{\OPT}{\operatorname{\rm OPT}}
\def\E{\operatorname*{\mathbb{E}}}

\def\poly{\mathop{\rm{poly}}\nolimits}
\def\opt{\mathop{\rm{opt}}\nolimits}

\newcommand{\Gen}{\operatorname{Gen}}
\newcommand{\Trace}{\operatorname{Trace}}

\newcommand{\zo}{\{0,1\}}

\title{Differentially Private Release and Learning of Threshold Functions\thanks{Earlier versions of this paper appeared in FOCS 2015 \cite{BunNSV15-FOCS} and was posted as  arXiv:1504.07553 [cs.CR].  This version corrects an error in the proof and statement of Theorem~\ref{thm:sanitization-vs-range}.}\stoctext{\\(Extended Abstract)}}
\author{	Mark Bun\thanks{School of Engineering \& Applied Sciences, Harvard University. \texttt{mbun@seas.harvard.edu}, \texttt{http://seas.harvard.edu/\textasciitilde mbun}. Supported by an NDSEG fellowship and NSF grant CNS-1237235.}\qquad
		Kobbi Nissim\thanks{Department of Computer Science, Georgetown University. \texttt{kobbi.nissim@georgetown.edu}. Work done when K.N.\ was visiting the Center for Research on Computation \& Society, Harvard University. Supported by NSF grant CNS-1237235, a gift from Google, Inc., and a Simons Investigator grant.} \qquad
		Uri Stemmer\thanks{Tel Aviv University. \texttt{u@uri.co.il}. Work done while U.S.\ was at Ben-Gurion University, supported by the Ministry of Science and Technology (Israel), by the Check Point Institute for Information Security, by the IBM PhD Fellowship Awards Program, and by the Frankel Center for Computer Science.} \qquad
		Salil Vadhan\thanks{Center for Research on Computation \& Society, School of Engineering \& Applied Sciences, Harvard University. \texttt{salil@seas.harvard.edu}, \texttt{http://seas.harvard.edu/\textasciitilde salil}.  Supported by NSF grant CNS-1237235, a gift from Google, Inc., and a Simons Investigator grant.}
}

\date{December 19, 2024}

\begin{document}

\maketitle

\begin{abstract}
We prove new upper and lower bounds on the sample complexity of $(\eps, \delta)$ differentially private algorithms for releasing approximate answers to threshold functions. A threshold function $c_x$ over a totally ordered domain $X$ evaluates to $c_x(y) = 1$ if $y \le x$, and evaluates to $0$ otherwise. We give the first nontrivial lower bound for releasing thresholds with $(\eps,\delta)$ differential privacy, showing that the task is impossible over an infinite domain $X$, and moreover requires sample complexity $n \ge \Omega(\log^*|X|)$, which grows with the size of the domain. Inspired by the techniques used to prove this lower bound, we give an algorithm for releasing thresholds with $n \le 2^{(1+ o(1))\log^*|X|}$ samples. This improves the previous best upper bound of $8^{(1 + o(1))\log^*|X|}$ (Beimel et al., RANDOM '13).

Our sample complexity upper and lower bounds also apply to the tasks of learning distributions with respect to Kolmogorov distance and of properly PAC learning thresholds with differential privacy. The lower bound gives the first separation between the sample complexity of properly learning a concept class with $(\eps,\delta)$ differential privacy and learning without privacy. For properly learning thresholds in $\ell$ dimensions, this lower bound extends to $n \ge \Omega(\ell \cdot \log^*|X|)$.

To obtain our results, we give reductions in both directions from releasing and properly learning thresholds and the simpler \emph{interior point problem}. Given a database $D$ of elements from $X$, the interior point problem asks for an element between the smallest and largest elements in $D$. We introduce new recursive constructions for bounding the sample complexity of the interior point problem, as well as further reductions and techniques for proving impossibility results for other basic problems in differential privacy.
\end{abstract}
\thispagestyle{empty}

\vfill


\newpage

\setcounter{page}{1}

\section{Introduction}

The line of work on {\em differential privacy}~\cite{DworkMcNiSm06} is aimed at
enabling useful statistical analyses on privacy-sensitive data while providing strong privacy protections for individual-level information.  Privacy is achieved in differentially private algorithms through randomization and the introduction of ``noise'' to obscure the effect of each individual, and thus differentially private algorithms can be less accurate than their non-private analogues.  Nevertheless,
by now a rich literature has shown that many data analysis tasks of interest are compatible with differential privacy, and generally the loss in accuracy vanishes as the number $n$ of individuals tends to infinity.  However, in many cases, there is still is a price of privacy hidden in these asymptotics --- in the rate at which the loss in accuracy vanishes, and in how large $n$ needs to be to start getting accurate results at all (the ``sample complexity'').

In this paper, we consider the price of privacy for three very basic types of computations involving threshold functions: query release, distribution learning with respect to Kolmogorov distance, and (proper) PAC learning.  In all cases, we show for the first time that accomplishing these tasks with differential privacy is {\em impossible} when the data universe is infinite (e.g. $\N$ or $[0,1]$) and in fact that the sample complexity must grow with the size $|X|$ of the data universe: $n=\Omega(\log^*|X|)$, 
which is tantalizingly close to the previous upper bound of $n=2^{O(\log^*|X|)}$~\cite{BeimelNiSt13b}.  We also provide simpler and somewhat improved upper bounds for these problems, reductions between these problems and other natural problems, as well as additional techniques that allow us to prove
impossibility results for infinite domains even when the sample complexity does not need to grow with the domain size (e.g. for PAC learning of point functions with ``pure'' differential privacy).

\subsection{Differential Privacy}

We recall the definition of differential privacy.  We think of a dataset as consisting of $n$ rows from a data universe $X$, where each row corresponds to one individual.  Differential privacy requires that no individual's data has a significant effect on the distribution of what we output.

\begin{definition}
A randomized algorithm $M : X^n\rightarrow Y$ is $(\eps,\delta)$ {\em differentially private} if for every two datasets $x,x'\in X^n$ that differ on one row, and every set $T\subseteq Y$, we have
$$\Pr[M(x)\in T]\leq e^{\eps}\cdot \Pr[M(x')\in T]+\delta.$$
\end{definition}

The original definition from \cite{DworkMcNiSm06} had $\delta=0$, and is sometimes referred to as {\em pure} differential privacy.  However, a number of subsequent works have shown that allowing a small (but negligible) value of $\delta$, referred to as {\em approximate differential privacy}, can provide substantial gains over  pure differential privacy~\cite{DworkLe09,HardtTa10,DworkRoVa10,De12,BeimelNiSt13b}.

The common setting of parameters is to take $\eps$ to be a small constant and $\delta$ to be negligible in $n$ (or a given security parameter).  To simplify the exposition, we fix $\eps=0.1$ and $\delta=1/n^{\log n}$ throughout the introduction (but precise dependencies on these parameters are given in the main body).

\subsection{Private Query Release}

Given a set $Q$ of queries $q : X^n\rightarrow \R$, the {\em query release} problem for $Q$ is to output accurate answers  to all queries in $Q$.  That is, we want a differentially private algorithm $M : X^n\rightarrow \R^{|Q|}$ such that for every dataset $D\in X^n$, with high probability over $y\leftarrow M(D)$, we have $|y_q-q(D)|\leq \alpha$ for all $q\in Q$, for an error parameter $\alpha$.

A special case of interest is the case where $Q$ consists of {\em counting queries}.  In this case, we are given a set $Q$ of predicates $q : X\rightarrow \{0,1\}$ on individual rows, and then extend them to databases by averaging. That is, $q(D) = (1/n)\sum_{i=1}^D q(D_i)$ counts the fraction of the population that satisfies predicate $q$.

The query release problem for counting queries is one of the most widely studied problems in differential privacy.  Early work on differential privacy implies that for every family of counting queries $Q$, the query release problem for $Q$ has sample complexity at most $\tO(\sqrt{|Q|})$ \cite{DinurNi03, DworkNi04, BlumDwMcNi05, DworkMcNiSm06}. That is, there is an $n_0=\tO(\sqrt{|Q|})$ such that for all $n\geq n_0$,
there is a differentially private mechanism
$M : X^n\rightarrow \R^{Q}$ that solves the query release problem for $Q$ with
error at most $\alpha=0.01$.
(Again, we treat $\alpha$ as a small constant to avoid an extra parameter in the introduction.)

Remarkably, Blum, Ligett, and Roth~\cite{BlumLiRo08} showed that if the data universe $X$ is finite, then the sample complexity grows much more slowly with $|Q|$ --- indeed the query release problem for $Q$ has sample complexity at most $O((\log |X|)(\log |Q|))$.  Hardt and Rothblum~\cite{HardtRo10} improved this bound to $\tO(\log |Q|\cdot \sqrt{\log |X|})$, which was recently shown to be optimal for
some families $Q$~\cite{BUV14}.

However, for specific query families of interest, the sample complexity can be significantly smaller.  In particular, consider the family of {\em point functions} over $X$, which is the family $\{q_x\}_{x\in X}$ where $q_x(y)$ is 1 iff $y=x$, and the family of {\em threshold functions} over $X$, where $q_x(y)$ is 1 iff $y\leq x$ (where $X$ is a totally ordered set).  The query release problems for these families correspond to the very natural tasks of producing $\ell_\infty$ approximations to the histogram and to the cumulative distribution function of the empirical data distribution, respectively.
For point functions, Beimel, Nissim, and Stemmer~\cite{BeimelNiSt13b} showed that the sample complexity has no dependence on $|X|$ (or $|Q|$, since
$|Q|=|X|$ for these families).  In the case of threshold functions, they showed that it has at most a very mild dependence on $|X|=|Q|$, namely $2^{O(\log^* |X|)}$.

Thus, the following basic questions remained open:
Are there differentially private algorithms for releasing threshold functions over an infinite data universe (such as $\N$ or $[0,1]$)?
If not, does the sample complexity for releasing threshold functions grow with the size $|X|$ of the data universe?

We resolve these questions:
\begin{theorem}
The sample complexity  of releasing threshold functions over a data universe $X$ with differential privacy is at least $\Omega(\log^* |X|)$.  In particular, there is no differentially private algorithm for releasing threshold functions over an infinite data universe.
\end{theorem}

In addition, inspired by the ideas in our lower bound, we present a simplification of the algorithm of \cite{BeimelNiSt13b} and improve the sample complexity to $2^{(1+o(1))\log^* |X|}$ (from roughly $8^{\log^* |X|}$).  Closing the gap between the lower bound of $\approx \log^* |X|$ and the upper bound of $\approx 2^{\log^* |X|}$ remains an intriguing open problem.

We remark that in the case of pure differential privacy ($\delta=0$), a sample complexity lower bound of $n=\Omega(\log |X|)$ follows from a standard ``packing argument''~\cite{HardtTa10,BeimelKaNi10}.  For point functions, this is matched by the standard Laplace mechanism~\cite{DworkMcNiSm06}.  For threshold functions, a matching upper bound was recently obtained~\cite{ReingoldRo14}, building on a construction of \cite{DworkNaPiRo10}.  We note that these algorithms have a slightly better dependence on the accuracy parameter $\alpha$ than our algorithm (linear rather than nearly linear in $1/\alpha$). In general, while packing arguments often yield tight lower bounds for pure differential privacy, they fail badly for approximate differential privacy, for which much less is known.

There is also a beautiful line of work on characterizing the $\ell_2$-accuracy achievable for query release in terms of other measures of the ``complexity'' of the family $Q$ (such as ``hereditary discrepancy'')~\cite{HardtTa10,BhaskaraDaKrTa12,MuthukrishnanNi12,NikolovTaZh13}.  However, the characterizations given in these works are tight only up to factors of
$\poly(\log |X|,\log |Q|)$ and thus do not give good estimates of the sample complexity (which is at most $(\log |X|)(\log |Q|)$ even for pure differential privacy, as mentioned above).

\subsection{Private Distribution Learning}

A fundamental problem in statistics is \emph{distribution learning}, which is the task of learning an unknown distribution $\cD$ given i.i.d.\ samples from it. The query release problem for threshold functions is closely related to the problem of learning an arbitrary distribution $\cD$ on $\R$ up to small error in Kolmogorov (or CDF) distance: Given $n$ i.i.d.\ samples $x_i \getsr \cD$, the goal of a distribution learner is to produce a CDF $F : X \to [0, 1]$ such that $|F(x) - F_{\cD}(x)| \le \alpha$ for all $x \in X$, where $\alpha$ is an accuracy parameter. While closeness in Kolmogorov distance is a relatively weak measure of closeness for distributions, under various structural assumptions (e.g. the two distributions have probability mass functions that cross in a constant number of locations), it implies closeness in the much stronger notion of total variation distance. Other works have developed additional techniques that use weak hypotheses learned under Kolmogorov distance to test and learn distributions under total variation distance (e.g. \cite{DaskalakisDiSeVaVa13, DaskalakisDiSe14, DaskalakisKa14}).

The well-known Dvoretzky-Kiefer-Wolfowitz inequality \cite{DvoretzkyKiWo56} implies that without privacy, any distribution over $X$ can be learned to within constant error with $O(1)$ samples. On the other hand, we show that with approximate differential privacy, the task of releasing thresholds is essentially equivalent to distribution learning. As a consequence, with approximate differential privacy, distribution learning instead requires sample complexity that grows with the size of the domain.
\begin{theorem} \label{thm:dist-informal}
The sample complexity of learning arbitrary distributions on a domain $X$ with differential privacy is at least $\Omega(\log^*|X|)$.
\end{theorem}
We prove Theorem \ref{thm:dist-informal} by showing that the problem of distribution learning with respect to Kolmogorov distance with differential privacy is essentially equivalent to query release for threshold functions. Indeed, query release of threshold functions amounts to approximating the \emph{empirical distribution} of a dataset with respect to Kolmogorov distance. Approximating the empirical distribution is of course trivial without privacy (since we are given it as input), but with privacy, it turns out to have essentially the same sample complexity as the usual distribution learning problem from i.i.d.\ samples.  More generally, query release for a family $Q$ of counting queries is equivalent to distribution learning with respect to the distance measure $d_{Q}(\cD, \cD') = \sup_{q \in Q} |\E[q(\cD)] - \E[q(\cD')]|$.

\subsection{Private PAC Learning}

Kasiviswanathan et al.~\cite{KasiviswanathanLeNiRaSm07} defined {\em private PAC learning} as a combination of probably approximately correct (PAC) learning~\cite{Valiant84} and differential privacy.
Recall that a PAC learning algorithm takes some $n$ labeled examples
$(x_i,c(x_i))\in X\times \zo$ where the $x_i$'s are i.i.d.\ samples of an arbitrary and unknown distribution on a data universe $X$ and $c : X\rightarrow \zo$ is an unknown {\em concept} from some concept class $C$.  The goal of the learning algorithm is
to output a {\em hypothesis} $h : X\rightarrow \zo$ that approximates $c$ well on the unknown distribution.   We are interested in PAC learning algorithms
$L : (X\times \zo)^n\rightarrow H$ that are also differentially private.  Here
$H$ is the {\em hypothesis class}; if $H \subseteq C$, then $L$ is called a {\em proper} learner.

As with query release and distribution learning, a natural problem is to characterize the {\em sample complexity} --- the minimum number $n$ of samples in order to achieve differentially private PAC learning for a given concept class $C$.
Without privacy,  it is well-known that
the sample complexity of (proper) PAC learning is proportional to the Vapnik--Chervonenkis (VC) dimension of the class $C$~\cite{VC,BlumerEhHaWa89,EHKV}.  In the initial work on differentially private learning, Kasiviswanathan et al.~\cite{KasiviswanathanLeNiRaSm07} showed that $O(\log|C|)$ labeled examples suffice for privately learning any concept class $C$.\footnote{As with the query release discussion, we omit the dependency on all parameters except for $|C|$, $|X|$ and $\VC(C)$.} The VC dimension of a concept class $C$ is always at most $\log|C|$, but is significantly lower for many interesting classes. Hence, the results of~\cite{KasiviswanathanLeNiRaSm07} left open the possibility that the sample complexity of private learning may be significantly higher than that of non-private learning.

In the case of {\em pure} differential privacy ($\delta=0$), this gap in the sample complexity was shown to be unavoidable in general.
Beimel, Kasiviswanathan, and Nissim~\cite{BeimelKaNi10}  considered the concept class $C$ of point functions over a data universe $X$, which have VC dimension 1 and hence can be (properly) learned without privacy with $O(1)$ samples. In contrast, they showed that proper PAC learning with pure differential privacy requires sample complexity $\Omega(\log |X|)=\Omega(\log |C|)$.
Feldman and Xiao~\cite{FX14} showed a similar separation even for improper learning --- the class $C$ of threshold functions over $X$ also has VC dimension 1, but PAC learning with pure differential privacy requires sample complexity $\Omega(\log |X|)=\Omega(\log |C|)$.

For approximate differential privacy ($\delta>0$), however, it was still open whether there is an asymptotic gap between the sample complexity of private learning and non-private learning.  Indeed, Beimel et al.~\cite{BeimelNiSt13b} showed that point functions
can be properly learned with approximate differential privacy using $O(1)$ samples (i.e. with no dependence on $|X|$).  For threshold functions, they exhibited a proper learner with sample complexity $2^{O(\log^* |X|)}$, but it was conceivable that the sample complexity could also be reduced to $O(1)$.

We prove that the sample complexity of proper PAC learning with approximate differential privacy can be asymptotically larger than the VC dimension:

\begin{theorem}
The sample complexity of properly learning threshold functions over a data universe $X$ with differential privacy is at least $\Omega(\log^* |X|)$.
\end{theorem}

This lower bound extends to the concept class of $\ell$-dimensional thresholds. An $\ell$-dimensional threshold function, defined over the domain $X^\ell$, is a conjunction of $\ell$ threshold functions, each defined on one component of the domain. This shows that our separation between the sample complexity of private and non-private learning applies to concept classes of every VC dimension.

\begin{theorem}
For every finite, totally ordered $X$ and $\ell\in \N$, the sample complexity of properly learning the class $C$ of $\ell$-dimensional threshold functions on $X^\ell$ with differential privacy is at least
$\Omega(\ell\cdot \log^*|X|) = \Omega(\VC(C)\cdot \log^* |X|)$.
\end{theorem}

Based on these results, it would be interesting to fully characterize the difference between the sample complexity of proper non-private learners and of proper learners with (approximate) differential privacy. Furthermore, our results still leave open the possibility that {\em improper} PAC learning with (approximate) differential privacy has sample complexity $O(\VC(C))$.  We consider this to be an important question for future work.

We also present a new result on improper learning of point functions with pure differential privacy over infinite countable domains.  Beimel et al.~\cite{BeimelKaNi10,BeimelNiSt13a} showed that for finite data universes $X$, the sample complexity of improperly learning point functions with pure differential
privacy does not grow with $|X|$.  They also gave a mechanism for learning point functions over infinite domains (e.g. $X=\N$), but the outputs of their mechanism do not have a finite description length (and hence cannot be implemented by an algorithm).  We prove that this is inherent:

\begin{theorem} \label{thm:pure-point-lb}
For every infinite domain $X$, countable hypothesis space $H$, and $n\in \N$, there is no (even improper) PAC learner $L : (X\times \zo)^n\rightarrow
H$ for point functions over $X$ that satisfies pure differential privacy.
\end{theorem}

\subsection{Techniques}

Our results for query release and proper learning of threshold functions are obtained by analyzing the sample complexity of a related but simpler problem, which
we call the {\em interior-point problem}.  Here we want a mechanism $M : X^n\rightarrow X$ (for a totally ordered data universe $X$) such that for every database $D\in X^n$, with high probability we have
$\min_i D_i \leq M(D)\leq \max_i D_i$.  We give reductions showing that the sample complexity of this problem is equivalent to the other ones we study:
\begin{theorem}
Over every totally ordered data universe $X$, the following four problems have the same sample complexity (up to constant factors) under differential privacy:
\stocrm{
\begin{enumerate}
\item The interior-point problem.
\item Query release for threshold functions.
\item Distribution learning (with respect to Kolmogorov distance).
\item Proper PAC learning of threshold functions.
\end{enumerate}
}
\stoctext{
(1) The interior-point problem, (2) Query release for threshold functions, (3) Distribution learning with respect to Kolmogorov distance, and (4) Proper PAC learning of threshold functions.
}
\end{theorem}
Thus we obtain our lower bounds and our simplified and improved upper bounds  for query release and proper learning by proving such bounds  for the interior-point problem, such as:

\begin{theorem} \label{thm:range-lb-informal}
The sample complexity for solving the interior-point problem over a data universe $X$ with differential privacy is $\Omega(\log^* |X|)$.
\end{theorem}

Note that for every fixed distribution $\cD$ over $X$ there exists a simple differentially private algorithm for solving the interior-point problem (w.h.p.) over databases sampled i.i.d.\ from $\cD$ -- simply output a point $z$ s.t. $\Pr_{x\sim \cD}[x\geq z]=1/2$. Hence, in order to prove Theorem~\ref{thm:range-lb-informal}, we show a (correlated) distribution $\cD$ over {\em databases} of size $n \approx \log^* |X|$ on which privately solving the interior-point problem is impossible. The construction is recursive: we use a hard distribution over databases of size $(n-1)$ over a data universe of size logarithmic in $|X|$ to construct
the hard distribution over databases of size $n$ over $X$.

By another reduction to the interior-point problem, we show an impossibility result for the following \emph{undominated-point} problem:
\begin{theorem}
For every $n\in \N$, there does not exist a differentially private mechanism $M : \N^n\rightarrow \N$ with the property that
for every dataset $D\in \N^n$, with high probability $M(D)\geq \min_i D_i$.
\end{theorem}
Note that for the above problem, one cannot hope to construct a single distribution over databases that every private mechanism fails on. The reason is that for any such distribution $\cD$, and any desired failure probability $\beta$, there is some number $K$ for which $\Pr_{D\sim\cD}[\max{D} > K] \le \beta$, and hence that the mechanism that always outputs $K$ solves the problem.
Hence, given a mechanism $\cM$ we must tailor a hard distribution $\cD_\cM$.
We use a similar mechanism-dependent approach to prove Theorem~\ref{thm:pure-point-lb}.

\subsection{Subsequent Work}

\paragraph{Private PAC learning.}
Our lower bound of $\Omega(\log^* |X|)$ on the sample complexity for privately learning threshold functions over a domain $X$ only holds for {\em proper} learning. Alon, Livni, Malliaris, and Moran~\cite{AlonLMM19} showed that the same lower bound also holds for {\em improper} private learners. Furthermore, they related the sample complexity privately learning of a general concept class $C$ to the {\em Littlestone dimension} of $C$, a combinatorial dimension that is known to characterize {\em online learnability} (non-privately).\footnote{Online learning is a learning model in which data becomes available in a sequential order and is used to update the predictor at each step.
The Littlestone dimension is a combinatorial 
parameter that characterizes online learning~\cite{Littlestone87online,Bendavid09agnostic}. The definition of this parameter uses the notion of {\it mistake-trees}:
these are binary decision trees whose internal nodes are labeled by data points. Any root-to-leaf path in a mistake tree can be described as a sequence of examples 
$(x_1,y_1),\ldots,(x_d,y_d)$, where $x_i$ is the data point of the $i$'th 
internal node in the path, and $y_i=1$ if the $(i+1)$'th node  
in the path is the right child of the $i$'th node, and otherwise $y_i = 0$.
We say that a tree $T$ is {\it shattered} by a class $C$ if for any root-to-leaf path
$(x_1,y_1),\ldots,(x_d,y_d)$ in $T$ there is $h\in C$ such that $h(x_i)=y_i$, for all $i\leq d$.
The Littlestone dimension of $C$, denoted by $\LDim(C)$, is the depth of the largest
complete tree that is shattered by~$C$.
} Specifically, Alon et al.~\cite{AlonLMM19} showed that if a class $C$ has Littlestone dimension $\LDim(C)$ then the sample complexity of privately learning $C$ is at least $\Omega(\log^* (\LDim(C)))$.
This follows from a classical result of Shelah~\cite{Shelah78classification}, who showed that every class of functions $C$ with a large Littlestone dimension has a large collection of 1-dimensional threshold functions embedded in it. 

Alon, Bun, Livni, Malliaris, and Moran~\cite{AlonBLMM22} then showed that the sample complexity of privately learning a class of functions $C$ can be {\em upper bounded} in terms of $\LDim(C)$. Together with the lower bound of Alon et al.~\cite{AlonLMM19}, this shows that the sample complexity of 
privately learning a class $C$
is {\em finite} if and only if the class $C$ has a finite Littlestone dimension. The generic upper bound of Alon et al.~\cite{AlonBLMM22} was then improved by Ghazi, Golowich, Kumar, and Manurangsi~\cite{GhaziG0M21}. Together, the results of~\cite{AlonLMM19,AlonBLMM22,GhaziG0M21} show that if a class $C$ has Littlestone dimension $\LDim(C)$ then the sample complexity of privately learning $C$ is at least $\Omega(\log^* (\LDim(C)))$ and at most $\tilde{O}(\LDim(C)^6)$. As the Littlestone dimension characterizes online learnability, these results imply a qualitative equivalence between online learnability and private PAC learnability.  

\paragraph{The private interior-point problem.} Followup works presented improved algorithms for the private interior-point problem over a domain $X$. Specifically, Kaplan, Ligett, Mansour, Naor, and Stemmer~\cite{KaplanLMNS20} presented an improved algorithm with sample complexity $\tilde{O}((\log^*|X|)^{1.5})$, as opposed to the upper bound of $2^{(1+o(1))\log^* |X|}$ presented in this work. Subsequently, 
Cohen, Lyu, Nelson, Sarl{\'{o}}s, and Stemmer~\cite{Cohen0NSS23}
presented a {\em tight} upper bound of $\tilde{\Theta}(\log^*|X|)$ for this problem (up to lower order terms). Via our reductions, this translate to near optimal algorithms for query release for threshold functions, distribution learning (with respect to Kolmogorov distance), and for private learning of threshold functions.

\stocrm{
\section{Preliminaries}

Throughout this work, we use the convention that $[n] = \{0, 1, \dots, n-1\}$ and write $\log$ for $\log_2$. We use $\thresh_X$ to denote the set of all {\em threshold functions} over a totally ordered domain $X$. That is,
$$\thresh\nolimits_X = \{c_x : x \in X\} \quad \text{where} \quad c_x(y) = 1 \text{ iff } y \le x.$$
\mnote{Didn't change this to $y \ge x$ (yet)}

\subsection{Differential Privacy} \label{sec:dp-prelim}

%

Our algorithms and reductions rely on a number of basic results about differential privacy. Early work on differential privacy showed how to solve the query release problem by adding independent Laplace noise to each exact query answer. A real-valued random variable is distributed as $\Lap(b)$ if its probability density function is $f(x)=\frac{1}{2b}\exp(-\frac{|x|}{b})$. We say a function $f:X^n \rightarrow \R^m$ has {\em sensitivity $\Delta$} if for all neighboring $D,D'\in X^n$, it holds that $||f(D)-f(D')||_1\leq \Delta$.

\begin{theorem}[The Laplace Mechanism \cite{DworkMcNiSm06}]\label{thm:lap}
Let $f:X^n \rightarrow \R^n$ be a sensitivity $\Delta$ function. The mechanism $A$ that on input $D\in X^n$
adds independent noise with distribution $\Lap(\Delta/\epsilon)$ to each coordinate of $f(D)$ preserves $\epsilon$-differential privacy.
\end{theorem}

We will present algorithms that access their input database using (several) differentially private mechanisms and use the following composition theorem to prove their overall privacy guarantee.

\begin{lemma}[Composition, e.g. \cite{DworkLe09}] \label{lem:composition}
Let $\cM_1: X^n \to \cR_1$ be $(\eps_1, \delta_1)$-differentially private. Let $\cM_2: X^n \times \cR_1 \to \cR_2$ be $(\eps_2, \delta_2)$-differentially private for any fixed value of its second argument. Then the composition $M(D) = \cM_2(D, \cM_1(D))$ is $(\eps_1 + \eps_2, \delta_1 + \delta_2)$-differentially private.
\end{lemma}

}


\section{The Interior Point Problem}

\subsection{Definition}
In this work we exhibit a close connection between the problems of privately learning and releasing threshold queries, distribution learning, and solving the \emph{interior point problem} as defined below.

\begin{definition}
An algorithm $A:X^n \to X$ solves the \emph{interior point problem} on $X$ with error probability $\beta$ if for every $D \in X^n$,
\[\Pr[\min D \le A(D) \le \max D] \ge 1 - \beta,\]
where the probability is taken over the coins of $A$. The sample complexity of the algorithm $A$ is the database size $n$.
\end{definition}

We call a solution $x$ with $\min D \le x \le \max D$ an \emph{interior point} of $D$. Note that $x$ need not be a member of the database $D$.

\stocrm{\subsection{Lower Bound}}
\stoctext{\subsection{Lower and Upper Bounds}}
\label{sec:lowerBound}

We now prove our lower bound on the sample complexity of private algorithms for solving the interior point problem.

\begin{theorem} \label{thm:range-lb}
Fix any constant $0 < \eps < 1/4$. Let $\delta(n) \le 1/(50 n^2)$. Then for every positive integer $n$, solving the interior point problem on $X$ with probability at least $3/4$ and with $(\eps, \delta(n))$-differential privacy requires sample complexity $n \ge \Omega(\log^*|X|)$.
\end{theorem}

Our choice of $\delta = O( 1/n^2 )$ is unimportant; any monotonically non-increasing 
convergent series will do. To prove the theorem, we inductively construct a sequence of database distributions $\{\cD_n\}$ supported on data universes $[S(n)]$ (for $S(n+1) = 2^{\tilde{O}(S(n))}$) over which any differentially private mechanism using $n$ samples must fail to solve the interior point problem. Given a hard distribution $\cD_n$ over $n$ elements $(x_1, x_2, \dots, x_n)$ from $[S(n)]$, we construct a hard distribution $\cD_{n+1}$ over elements $(y_0, y_1, \dots, y_n)$ from $[S(n+1)]$ by setting $y_0$ to be a random number, and letting each other $y_i$ agree with $y_0$ on the $x_i$ most significant digits. We then show that if $y$ is the output of any differentially private interior point mechanism on $(y_0, \dots, y_n)$, then with high probability, $y$ agrees with $y_0$ on at least $\min x_i$ entries and at most $\max x_i$ entries. Thus, a private mechanism for solving the interior point problem on $\cD_{n+1}$ can be used to construct a private mechanism for $\cD_n$, and so $\cD_{n+1}$ must also be a hard distribution.

The inductive lemma we prove depends on a number of parameters we now define. Fix $\frac{1}{4}>\eps, \beta > 0$. Let $\delta(n)$ be any positive non-increasing
sequence for which
\[P_n \triangleq \frac{e^\eps}{e^\eps + 1} + (e^\eps+1)\sum_{j=1}^{n} \delta(j) \le 1 - \beta \] for every $n$.
In particular, it suffices that
\[\sum_{n=1}^\infty \delta(n) \leq \frac{\frac{1}{3}-\beta}{e^\epsilon + 1}.\]
Let $b(n) = 1/\delta(n)$ and define the function $S$ recursively by
\[S(1) = 2 \quad \text{and} \quad S(n+1) = b(n)^{S(n)}.\]

\begin{lemma} \label{lem:range-lb-dist}
For every positive integer $n$, there exists a distribution $\cD_n$ over databases $D \in [S(n)]^n = \{0, 1, \dots, S(n)-1\}^n$ such that for every $(\eps, \delta(n))$-differentially private mechanism $\cM$,
\[\Pr[\min D \le \cM(D) \le \max D] \le P_n,\]
where the probability is taken over $D \getsr \cD_n$ and the coins of $\cM$.
\end{lemma}

In this section, we give a direct proof of the lemma and in \stocrm{Appendix \ref{app:range-fpc}}\stoctext{the full version of this paper}, we show how the lemma follows from the construction of a new combinatorial object we call an ``interior point fingerprinting code.'' This is a variant on traditional fingerprinting codes, which have been used recently to show nearly optimal lower bounds for other problems in approximate differential privacy \cite{BUV14, DworkTaThZh14, BassilySmTh14}.

\begin{proof}
The proof is by induction on $n$. We first argue that the claim holds for $n = 1$ by letting $\cD_1$ be uniform over the singleton databases $(0)$ and $(1)$. To that end let $x \getsr\cD_1$ and note that for any $(\eps, \delta(1))$-differentially private mechanism $\cM_0:\{0,1\}\rightarrow\{0,1\}$ it holds that
$$
\Pr[\cM_0(x)=x]\leq e^\eps \Pr[\cM_0(\bar{x})=x]+\delta(1)=e^\eps (1-\Pr[\cM_0(x)=x])+\delta(1),
$$
giving the desired bound on $\Pr[\cM_0(x)=x]$.

Now inductively suppose we have a distribution $\cD_n$ that satisfies the claim. We construct a distribution $\cD_{n+1}$ on databases $(y_0, y_1, \dots, y_n) \in [S(n+1)]^{n+1}$ that is sampled as follows:
\begin{itemize}
\item Sample $(x_1, \dots, x_n) \getsr \cD_n$.
\item Sample a uniformly random $y_0 \getsr [S(n+1)]$. We write the base $b(n)$ representation of $y_0$ as $y_0^{(1)}y_0^{(2)}\dots y_0^{(S(n))}$.
\item For each $i = 1, \dots, n$ let $y_i$ be a base $b(n)$ number (written $y_i^{(1)}y_i^{(2)}\dots y_i^{(S(n))}$) that agrees with the base $b(n)$ representation of $y_0$ on the first $x_i$ digits and contains a random sample from $[b(n)]$ in every index thereafter.
\end{itemize}
Suppose for the sake of contradiction that there were an $(\eps, \delta(n+1))$-differentially private mechanism $\hat{\cM}$ that could solve the interior point problem on $\cD_{n+1}$ with probability greater than $P_{n+1}$. We use $\hat{\cM}$ to construct the following private mechanism $\cM$ for solving the interior point problem on $\cD_n$, giving the desired contradiction:

\begin{algorithm}[H]
\caption{$\cM(D)$}
{\bf Input:} Database $D = (x_1, \dots, x_n) \in [S(n)]^n$
\begin{enumerate}
\item Construct $\hat{D} = (y_0, \dots, y_n)$ by sampling from $\cD_{n+1}$, but starting with the database $D$. That is, sample $y_0$ uniformly at random and set every other $y_i$ to be a random base $b(n)$ string that agrees with $y_0$ on the first $x_i$ digits.
\item Compute $y \getsr \hat{\cM}(\hat{D})$.
\item Return the length of the longest prefix of $y$ (in base $b(n)$ notation) that agrees with $y_0$.
\end{enumerate}
\end{algorithm}

The mechanism $\cM$ is also $(\eps, \delta(n+1))$-differentially private, since for all pairs of adjacent databases $D\sim D'$ and every $T \subseteq [S(n)]$,

\begin{align*}
\Pr[\cM(D) \in T] &= \E_{y_0 \getsr [S(n+1)]} \Pr[\hat{\cM}(\hat{D}) \in \hat{T} \mid y_0] \\
&\le \E_{y_0 \getsr [S(n+1)]} (e^\eps \Pr[\hat{\cM}(\hat{D}') \in \hat{T} \mid y_0] + \delta) & \text{since $\hat{D} \sim \hat{D}'$ for fixed $y_0$}\\
& = e^{\eps} \Pr[\cM(D') \in T] + \delta,
\end{align*}
where $\hat{T}$ is the set of $y$ that agree with $y_0$ in exactly the first $x$ entries for some $x \in T$.

Now we argue that $\cM$ solves the interior point problem on $\cD_n$ with probability greater than $P_n$. First we show that $x \ge \min D$ with probability greater than $P_{n+1}$. Observe that by construction, all the elements of $\hat{D}$ agree in at least the first $\min D$ digits, and hence so does any interior point of $\hat{D}$. Therefore, if $\cM'$ succeeds in outputting an interior point $y$ of $\hat{D}$, then $y$ must in particular agree with $y_0$ in at least $\min D$ digits, so $x \ge \min D$.

Now we use the privacy that $\hat{\cM}$ provides to $y_0$ to show that $x \le \max D$ except with probability at most $e^{\eps}/b(n)+ \delta(n+1)$. Fix a database $D$. Let $w = \max D$, and fix all the randomness of $\cM$ but the $(w+1)$st entry of $y_0$ (note that since $w = \max D$, this fixes $y_1,\ldots,y_n$). Since the $(w+1)$st entry of $y_0$ is still a uniformly random element of $[b(n)]$, the privately produced entry $y^{w+1}$ should not be able to do much better than randomly guessing $y_0^{(w+1)}$. Formally, for each $z \in [b(n)]$, let $\hat{D}_z$ denote the database $\hat{D}$ with $y_0^{(w+1)}$ set to $z$ and everything else fixed as above. Then by the differential privacy of $\hat{\cM}$,

\begin{align*}
\Pr_{z \in [b(n)]}[\hat{\cM}(\hat{D}_z)^{w+1} = z] &= \frac{1}{b(n)} \sum_{z \in [b(n)]} \Pr[\hat{\cM}(\hat{D}_z)^{w+1} = z] \\
&\le \frac{1}{b(n)} \sum_{z \in [b(n)]} \E_{z' \getsr [b(n)]} \left[e^\eps \Pr[\hat{\cM}(\hat{D}_{z'})^{w+1} = z] +\delta(n+1)\right] \\
&\le \frac{e^{\eps}}{b(n)} + \delta(n+1),
\end{align*}
where all probabilities are also taken over the coins of $\hat{\cM}$. Thus $x \le \max D$ except with probability at most $e^{\eps}/b(n)+ \delta(n+1)$. By a union bound, $\min D\le x \le \max D$ with probability greater than
\[P_{n+1} -\left( \frac{e^{\eps}}{b(n)} + \delta(n+1)\right) \geq P_n.\]
This gives the desired contradiction.
\end{proof}

\stocrm{
We now prove Theorem \ref{thm:range-lb} by estimating the $S(n)$ guaranteed by Lemma \ref{lem:range-lb-dist}.

\begin{proof}[Proof of Theorem \ref{thm:range-lb}]
Let $S(n)$ be as in Lemma \ref{lem:range-lb-dist}. We introduce the following notation for iterated exponentials:
\[{\rm tower}^{(0)}(x) = x \quad \text{and} \quad  {\rm tower}^{(k)}(x) = 2^{{\rm tower}^{(k-1)}(x)}.\]
Observe that for $k \ge 1$, $x > 0$, and $M > 16$,
\begin{align*}
M^{{\rm tower}^{(k)}(x)} &= 2^{{\rm tower}^{(k)}(x) \log M} \\
&= {\rm tower}^{(2)}({\rm tower}^{(k-1)}(x) + \log\log M) \\
&\le {\rm tower}^{(2)}({\rm tower}^{(k-1)}(x + \log\log M)) \\
&= {\rm tower}^{(k+1)}(x + \log\log M).
\end{align*}
By induction on $n$ we get an upper bound of
\[S(n+1) \le {\rm tower}^{(n)}(2 + n\log\log (cn^2)) \le {\rm tower}^{(n + \log^*(cn^2))}(1).\]
This immediately shows that solving the interior point problem on $X=[S(n)]$ requires sample complexity
\begin{align*}
n &\ge \log^*S(n) - \log^*(cn^2)\\
 &\ge \log^*S(n) - O(\log^*\log^*S(n))\\
&= \log^*|X| - O(\log^*\log^*|X|).
\end{align*}
To get a lower bound for solving the interior point problem on $X$ when $|X|$ is not of the form $S(n)$, note that a mechanism for $X$ is also a mechanism for every $X'$ s.t.\ $|X'|\le |X|$. The lower bound follows by setting $|X'|=S(n)$ for the largest $n$ such that $S(n)\leq|X|$.
\end{proof}
}

\stoctext{The proof of Theorem \ref{thm:range-lb} follows by estimating the $S(n)$ guaranteed by Lemma \ref{lem:range-lb-dist}, and appears in the full version of this work. Using similar ideas, we also prove the following upper bound on the sample complexity of solving the interior point problem.

\begin{theorem}
Let $\beta,\epsilon,\delta > 0$ and let $X$ be a finite totally ordered domain.
There is an $(\epsilon,\delta)$-differentially private algorithm for solving the interior point problem on $X$ with failure probability $\beta$ and sample complexity $n=\frac{18500}{\epsilon} \cdot 2^{\log^*|X|} \cdot \log^*(|X|) \cdot \ln(\frac{4\log^*|X|}{\beta\epsilon\delta})$.
\end{theorem}
}

\def\RecPrefix{\mbox{\it RecPrefix}}

\stocrm{
\subsection{Upper Bound}
We now present a recursive algorithm, $\RecPrefix$, for privately solving the interior point problem.

\begin{theorem} \label{thm:range_queries_upperbound}
Let $\beta,\epsilon,\delta > 0$, let $X$ be a finite, totally ordered domain, and let $n \in \N$ with $n\geq\frac{18500}{\epsilon} \cdot 2^{\log^*|X|} \cdot \log^*(|X|) \cdot \ln(\frac{4\log^*|X|}{\beta\epsilon\delta})$. If $\RecPrefix$ (defined below) is executed on a database $S \in X^n$ with parameters $\frac{\beta}{3\log^*|X|},\frac{\epsilon}{2\log^*|X|},\frac{\delta}{2\log^*|X|}$, then
\begin{enumerate}
	\item $\RecPrefix$ is $(\epsilon,\delta)$-differentially private;
	\item With probability at least $(1-\beta)$, the output $x$ satisfies
	$\min\{x_i:x_i\in S\}\leq x\leq\max\{x_i:x_i\in S\}$.
\end{enumerate}
\end{theorem}

The idea of the algorithm is that on each level of recursion, $\RecPrefix$ takes an input database $S$ over $X$ and constructs a database $S'$ over a smaller universe $X'$, where $|X'| = \log |X|$, in which every element is the length of the longest prefix of a pair of elements in $S$ (represented in binary). In a sense, this reverses the construction presented in Section~\ref{sec:lowerBound}.

\subsubsection{The exponential and choosing mechanisms}

Before formally presenting the algorithm $\RecPrefix$, we introduce several additional algorithmic tools. One primitive we will use is the exponential mechanism of McSherry and Talwar~\cite{McSherryTa07}. Let $X^*$ denote the set of all finite databases over the universe $X$. A quality function  $q:X^*\times \cF \rightarrow \N$ defines an {\em optimization problem} over the domain $X$ and a finite solution set $\cF$:  Given a database $S \in X^*$, choose $f\in\FFF$ that (approximately) maximizes $q(S,f)$. The exponential mechanism solves such an optimization problem by choosing a random solution where the probability of outputting any solution $f$ increases exponentially with its quality $q(D,f)$. Specifically, it outputs each $f \in \cF$ with probability $\propto \exp\left(\epsilon \cdot q(S,f) /2 \Delta q\right)$. Here, the sensitivity of a quality function, $\Delta q$, is the maximum over all $f\in\cF$ of the sensitivity of the function $q(\cdot, f)$.

\begin{proposition}[Properties of the Exponential Mechanism] \label{prop:exp_mech}
\ \begin{enumerate}
\item The exponential mechanism is $\epsilon$-differentially private.
\item Let $q$ be a quality function with sensitivity at most $1$. Fix a database $S \in X^n$ and let $\OPT = \max_{f\in \cF}\{q(S,f)\}$. Let $t >0$. Then exponential mechanism outputs a solution $f$ with $q(S,f)\leq \OPT - tn$ with probability at most $|\cF| \cdot \exp(-\epsilon tn /2)$.
\end{enumerate}
\end{proposition}

We will also use an $(\eps, \delta)$-differentially private variant of the exponential mechanism called the \emph{choosing mechanism}, introduced in \cite{BeimelNiSt13b}.

A quality function with sensitivity at most $1$ is of {\em $k$-bounded-growth} if adding an element to a database can increase (by 1) the score of at most $k$ solutions,\mnote{We do mean ``adding an element'' here} without changing the scores of other solutions.
\stocrm{Specifically, it holds that
\begin{enumerate}
\item $q(\emptyset, f) = 0$ for all $f \in \cF$,
\item If $S_2 = S_1 \cup \{x\}$, then $q(S_1, f) + 1 \ge q(S_2, f) \ge q(S_1, f)$ for all $f \in \cF$, and
\item There are at most $k$ values of $f$ for which $q(S_2, f) = q(S_1, f) + 1$.
\end{enumerate}
}

The choosing mechanism is a differentially private algorithm for approximately solving bounded-growth choice problems. Step~1 of the algorithm checks whether a good solution exists (otherwise any solution is approximately optimal) and Step~2 invokes the exponential mechanism, but with the {\em small} set $G(S)$ of good solutions instead of $\FFF$.

\begin{algorithm}
\caption{Choosing Mechanism}
\label{alg:choosing}
{\bf Input:} database $S$, quality function $q$, solution set $\FFF$, and parameters $\beta,\epsilon,\delta$ and $k$.
\begin{enumerate}[topsep=-1pt, rightmargin=10pt]
\item Set $\widetilde{\OPT}=\OPT+\Lap(\frac{4}{\epsilon})$. If $\widetilde{\OPT}<\frac{8}{\epsilon}\ln(\frac{4k}{\beta\epsilon\delta})$ then halt and return $\bot$.
\item Let $G(S)=\{f\in\FFF : q(S,f)\geq1\}$. Choose and return $f\in G(S)$ using the exponential mechanism with parameter $\frac{\epsilon}{2}$.\\
\end{enumerate}
\end{algorithm}

The following lemmas give the privacy and utility guarantees of the choosing mechanism. We give a slightly improved utility result over \cite{BeimelNiSt13b}, and the analysis is presented in Appendix \ref{app:choosing}.

\begin{lemma} \label{lem:CMprivacy}
Fix $\delta > 0$, and $0 < \epsilon\leq2$.
If $q$ is a $k$-bounded-growth quality function, then the choosing mechanism is $(\epsilon,\delta)$-differentially private.
\end{lemma}

\begin{lemma}\label{lem:CMscoreOne}
Let the choosing mechanism be executed on a $k$-bounded-growth quality function, and on a database $S$ s.t. there exists a solution $\hat{f}$ with quality $q(S,\hat{f})\geq\frac{16}{\epsilon}\ln(\frac{4k}{\beta\epsilon\delta})$.
With probability at least $(1-\beta)$, the choosing mechanism outputs a solution $f$ with quality
$q(S,f)\geq1$.
\end{lemma}

\begin{lemma}\label{lem:CMutility}
Let the choosing mechanism be executed on a $k$-bounded-growth quality function, and on a database $S$ containing $m$ elements.
With probability at least $(1-\beta)$, the choosing mechanism outputs a solution $f$ with quality
$q(S,f)\geq\OPT-\frac{16}{\epsilon}\ln(\frac{4km}{\beta\epsilon\delta})$.
\end{lemma}

\subsubsection{The $\RecPrefix$ algorithm}

We are now ready to present and analyze the algorithm $\RecPrefix$.

\begin{algorithm}[H]
\caption{$\RecPrefix$}
{\bf Input:} Database $S=(x_j)_{j=1}^n \in X^n$, parameters $\beta,\epsilon,\delta$.
\begin{enumerate}[rightmargin=10pt,itemsep=1pt]
\item If $|X|\leq32$, then use the exponential mechanism with privacy parameter $\epsilon$ and quality function $q(S, x)= \min\left\{ \#\{j : x_j \ge x\}, \#\{j : x_j \le x\}\right\} $ to choose and return a point $x\in X$.

\item Let $k= \lfloor \frac{386}{\epsilon}\ln(\frac{4}{\beta\epsilon\delta}) \rfloor$, and let $Y=(y_1,y_2,\ldots,y_{n-2k})$ be a random permutation of the smallest $(n{-}2k)$ elements in $S$.

\item For $j=1$ to $\frac{n-2k}{2}$, define $z_j$ as the length of the longest prefix for which $y_{2j-1}$ agrees with $y_{2j}$ (in base 2 notation).

\item Execute $\RecPrefix$ recursively on $S'=(z_j)_{j=1}^{(n-2k)/2}\in (X')^{(n-2k)/2}$ with parameters $\beta,\epsilon,\delta$. Recall that $|X'| = \log |X|$. Denote the returned value by $z$.

\item Use the choosing mechanism to choose a prefix $L$ of length $(z+1)$ with a large number of agreements among elements in $S$.
Use parameters $\beta,\epsilon,\delta$, and the quality function
$q:X^* \times X^{z+1}\rightarrow\N$, where $q(S, I)$ is the number of agreements on $I$ among $x_1,\ldots,x_n$.

\item For $\sigma\in\{0,1\}$, define $L_{\sigma}\in X$ to be the prefix $L$ followed by $(\log|X|-z-1)$ appearances of $\sigma$.

\item Compute $\widehat{big}=\Lap(\frac{1}{\epsilon})+ \#\{j : x_j\geq L_1\}$.

\item If $\widehat{big}\geq\frac{3k}{2}$ then return $L_1$. Otherwise return $L_0$.
\end{enumerate}
\end{algorithm}

We start the analysis of $\RecPrefix$ with the following two simple observations.

\begin{observation}\label{obs:RecPrefixIterations}
There are at most $\log^*|X|$ recursive calls throughout the execution of $\RecPrefix$ on a database $S\in X^*$.
\end{observation}

\begin{observation}\label{obs:databaseSize}
Let $\RecPrefix$ be executed on a database $S\in X^n$, where $n\geq 2^{\log^*|X|} \cdot \frac{2312}{\epsilon} \cdot \ln(\frac{4}{\beta\epsilon\delta})$.
Every recursive call throughout the execution operates on a database containing at least $\frac{1540}{\epsilon} \cdot \ln(\frac{4}{\beta\epsilon\delta})$ elements.
\end{observation}

\begin{proof}
This follows from Observation~\ref{obs:RecPrefixIterations} and from the fact that the $i^\text{th}$ recursive call is executed on a database of size $n_i=\frac{n}{2^{i-1}}-k\sum_{\ell=0}^{i-2}(\frac{1}{2})^{\ell} \geq \frac{n}{2^i} - 2k$.
\end{proof}

We now analyze the utility guarantees of $\RecPrefix$ by proving the following lemma.

\begin{lemma}\label{lem:RecPrefixUtility}
Let $\beta,\epsilon,\delta$, and $S\in X^n$ be inputs on which $\RecPrefix$ performs at most $N$ recursive calls,
all of which are on databases of at least $\frac{1540}{\epsilon}\cdot\ln(\frac{4}{\beta\epsilon\delta})$ elements.
With probability at least $(1-3\beta N)$, the output $x$ is s.t.
\begin{enumerate}
\item $\exists x_i\in S$ s.t. $x_i\leq x$;
\item $|\{i: x_i\geq x\}|\geq k\triangleq  \lfloor \frac{386}{\epsilon}\cdot\ln(\frac{4}{\beta\epsilon\delta}) \rfloor$.
\end{enumerate}
\end{lemma}

Before proving the lemma, we make a combinatorial observation that motivates the random shuffling in Step~2 of $\RecPrefix$. A pair of elements  $y,y'\in S$ is useful in Algorithm $\RecPrefix$ if many of the values in $S$ lie between $y$ and $y'$ -- a prefix on which $y,y'$ agree is also a prefix of every element between $y$ and $y'$. A prefix common to a useful pair can hence be identified privately via stability-based techniques. Towards creating useful pairs, the set $S$ is shuffled randomly. We will use the following lemma:

\begin{claim}\label{claim:randomPermutation}
Let $(\Pi_1,\Pi_2,\ldots,\Pi_n)$ be a random permutation of $(1,2,\ldots,n)$. Then for all $r \ge 1$,
$$
\Pr\left[  \left| \left\{ i : \left| \Pi_{2i-1} - \Pi_{2i} \right| \leq \frac{r}{12} \right\} \right| \geq r  \right] \leq 2^{-r}
$$
\end{claim}

\begin{proof}
We need to show that w.h.p.\ there are at most $r$ ``bad'' pairs $(\Pi_{2i-1}, \Pi_{2i})$ within distance $\frac{r}{12}$. For each $i$, we call $\Pi_{2i-1}$ the left side of the pair, and $\Pi_{2i}$ the right side of the pair.
Let us first choose $r$ elements to be placed on the left side of $r$ bad pairs (there are ${n \choose r}$ such choices).
Once those are fixed, there are at most $(\frac{r}{6})^r$ choices for placing elements on the right side of those pairs.
Now we have $r$ pairs and $n-2r$ unpaired elements that can be shuffled in $(n-r)!$ ways.
Overall, the probability of having at least $r$ bad pairs is at most
$$
\frac{{n \choose r}(\frac{r}{6})^r(n-r)!}{n!}=\frac{(\frac{r}{6})^r}{r!}\leq\frac{(\frac{r}{6})^r}{\sqrt{r}r^r e^{-r}}=\frac{e^r}{\sqrt{r}6^r}\leq2^{-r},
$$
where we have used Stirling's approximation for the first inequality.
\end{proof}

Suppose we have paired random elements in our input database $S$, and constructed a database $S'$ containing lengths of the prefixes for those pairs.
Moreover, assume that by recursion we have identified a length $z$ which is the length at least $r$ random pairs.
Although those prefixes may be different for each pair, Claim~\ref{claim:randomPermutation} guarantees that (w.h.p.) at least one of these prefixes is the prefix of at least $\frac{r}{12}$ input elements.
This will help us in (privately) identifying such a prefix.

\begin{proof}[Proof of Lemma \ref{lem:RecPrefixUtility}]
The proof is by induction on the number of recursive calls, denoted as $t$.
For $t=1$ (i.e., $|X|\leq32$), the claim holds as long as the exponential mechanism outputs an $x$ with $q(S, x) \ge k$ except with probability at most $\beta$. By Proposition~\ref{prop:exp_mech}, it suffices to have $n\geq\frac{1540}{\epsilon} \cdot \ln(\frac{4}{\beta\epsilon\delta})$, since $32 \exp(-\eps (n/2 - k) / 2) \le \beta$.

Assume that the stated lemma holds whenever $\RecPrefix$ performs at most $t-1$ recursive calls.
Let $\beta,\epsilon,\delta$ and $S=(x_i)_{i=1}^n\in X^n$ be inputs on which algorithm $\RecPrefix$ performs $t$ recursive calls, all of which are on databases containing at least $\frac{1540}{\epsilon} \cdot\ln(\frac{4}{\beta\epsilon\delta})$ elements.
Consider the first call in the execution on those inputs, and let $y_1,\ldots,y_{n-2k}$ be the random permutation chosen on Step~2.
We say that a pair $y_{2j-1},y_{2j}$ is {\em close} if $$\left|i :
\begin{array}{c}
	y_{2j-1}\leq y_i \leq y_{2j}\\
	\text{or}\\
	y_{2j}\leq y_i \leq y_{2j-1}
\end{array} \right|\leq\frac{k-1}{12}.$$
By Claim~\ref{claim:randomPermutation}, except with probability at most $2^{-(k-1)} < \beta$, there are at most
$(k-1)$ close pairs.
We continue the proof assuming that this is the case.

Let $S'=(z_i)_{i=1}^{(n-2k)/2}$ be the database constructed in Step~3.
By the inductive assumption, with probability at least $(1-3\beta(t-1))$, the value $z$ obtained in Step~4 is s.t.
(1) $\exists z_i\in S'$ s.t. $z_i\leq z$; and (2) $|\{z_i\in S' : z_i\geq z\}|\geq k$.
We proceed with the analysis assuming that this event happened.

By (2), there are at least $k$ pairs $y_{2j-1},y_{2j}$ that agree on a prefix of length at least $z$.
At least one of those pairs, say $y_{2j^*-1},y_{2j^*}$, is not {\em close}.
Note that every $y$ between $y_{2j^*-1}$ and $y_{2j^*}$ agrees on the same prefix of length $z$, and that there are at least $\frac{k-1}{12}$ such elements in $S$.
Moreover, as the next bit is either 0 or 1, at least half of those elements agree on a prefix of length $(z+1)$.
Thus, when using the choosing mechanism on Step~5 (to choose a prefix of length $(z+1)$), there exists at least one prefix with quality at least $\frac{k-1}{24}\ge\frac{16}{\epsilon}\cdot\ln(\frac{4}{\beta\epsilon\delta})$. By Lemma \ref{lem:CMutility}, the choosing mechanism ensures, therefore, that with probability at least $(1-\beta)$, the chosen prefix $L$ is the prefix of at least one $y_{i'}\in S$, and, hence, this $y_{i'}$ satisfies $L_0\leq y_{i'}\leq L_1$ (defined in Step~6). We proceed with the analysis assuming that this is the case.

Let $z_{\hat{j}}\in S'$ be s.t. $z_{\hat{j}}\leq z$. By the definition of $z_{\hat{j}}$, this means that $y_{2\hat{j}-1}$ and $y_{2\hat{j}}$ agree on a prefix of length at most $z$. Hence, as $L$ is of length $z+1$, we have that either $\min\{ y_{2\hat{j}-1},y_{2\hat{j}}\}<L_0$ or $\max\{ y_{2\hat{j}-1},y_{2\hat{j}}\}>L_1$.
If $\min\{ y_{2\hat{j}-1},y_{2\hat{j}}\}<L_0$, then $L_0$ satisfies Condition~1 of being a good output. It also satisfies Condition~2 because $y_{i'} \ge L_0$ and $y_{i'} \in Y$, which we took to be the smallest $n-2k$ elements of $S$. Similarly, $L_1$ is a good output if $\max\{ y_{2\hat{j}-1},y_{2\hat{j}}\}>L_1$.
In any case, at least one out of $L_0,L_1$ is a good output.

If both $L_0$ and $L_1$ are good outputs, then Step~8 cannot fail.
We have already established the existence of $L_0\leq y_{i'}\leq L_1$. Hence, if $L_1$ is not a good output, then there are at most $(k{-}1)$ elements $x_i\in S$ s.t. $x_i\geq L_1$. Hence, the probability of $\widehat{big} \ge 3k/2$ and Step~8 failing is at most $\exp(-\frac{\epsilon k}{2})\leq\beta$. It remains to analyze the case where $L_0$ is not a good output (and $L_1$ is).

If $L_0$ is not a good output, then every $x_j\in S$ satisfies $x_j>L_0$.
In particular, $\min\{ y_{2\hat{j}-1},y_{2\hat{j}}\}>L_0$, and, hence, $\max\{ y_{2\hat{j}-1},y_{2\hat{j}}\}>L_1$.
Recall that there are at least $2k$ elements in $S$ which are bigger than $\max\{ y_{2\hat{j}-1},y_{2\hat{j}}\}$.
As $k\geq\frac{2}{\epsilon}\ln(\frac{1}{\beta})$, the probability that $\widehat{big} < 3k/2$ and $\RecPrefix$ fails to return $L_1$ in this case is at most $\beta$.

All in all, $\RecPrefix$ fails to return an appropriate $x$ with probability at most $3\beta t$.
\end{proof}

We now proceed with the privacy analysis.

\begin{lemma}\label{lem:RecPrefixPrivacy}
When executed for $N$ recursive calls, $\RecPrefix$ is $(2\epsilon N,2\delta N)$-differentially private.
\end{lemma}

\begin{proof}
The proof is by induction on the number of recursive calls, denoted by $t$.
For $t=1$ (i.e., $|X|\leq32$), then by Proposition~\ref{prop:exp_mech} the exponential mechanism ensures that $\RecPrefix$ is $(\epsilon,0)$-differentially private.
Assume that the stated lemma holds whenever $\RecPrefix$ performs at most $t-1$ recursive calls, and let $S_1,S_2\in X^*$ be two neighboring databases on which $\RecPrefix$ performs $t$ recursive calls.\footnote{The recursion depth is determined by $|X|$, which is identical in $S_1$ and in $S_2$.}
Let $\BBB$ denote an algorithm consisting of steps~1-4 of $\RecPrefix$ (the output of $\BBB$ is the value $z$ from Step~4).
Consider the executions of $\BBB$ on $S_1$ and on $S_2$, and denote by $Y_1,S'_1$ and by $Y_2,S'_2$ the elements $Y,S'$ as they are in the executions on $S_1$ and on $S_2$.

We show that the distributions on the databases $S'_1$ and $S'_2$ are similar in the sense that for each database in one of the distributions
there exist a neighboring database in the other that have the same
probability. Thus, applying the recursion (which is differentially private by the inductive assumption) preserves privacy. We now make this argument formal.

First note that as $S_1,S_2$ differ in only one element,
there is a bijection between
orderings $\Pi$ and $\widehat{\Pi}$ of the smallest $(n-2k)$ elements of $S_1$ and of $S_2$ respectively s.t. $Y_1$ and $Y_2$ are neighboring databases. This is because there exists a permutation of the smallest $(n-2k)$ elements of $S_1$ that is a neighbor of the smallest $(n-2k)$ elements of $S_2$; composition with this fixed permutation yields the desired bijection.
Moreover, note that whenever $Y_1,Y_2$ are neighboring databases, the same is true for $S'_1$ and $S'_2$. Hence, for every set of outputs $F$ it holds that

\begin{eqnarray*}
\Pr[\BBB(S)\in F] &=& \sum_{\Pi}\Pr[\Pi]\cdot\Pr[\RecPrefix(S'_1) \in F|\Pi]\\
&\leq& e^{2\epsilon(t-1)}\cdot\sum_{\Pi}\Pr[\Pi]\cdot\Pr[\RecPrefix(S'_2) \in F|\widehat{\Pi}]+2\delta(t-1)\\
&=& e^{2\epsilon(t-1)}\cdot\sum_{\widehat{\Pi}}\Pr[\widehat{\Pi}]\cdot\Pr[\RecPrefix(S'_2) \in F|\widehat{\Pi}]+2\delta(t-1)\\
&=& e^{2\epsilon(t-1)}\cdot\Pr[\BBB(S')\in F]+2\delta(t-1)\\
\end{eqnarray*}

So when executed for $t$ recursive calls, the sequence of Steps~1-4 of $\RecPrefix$ is $(2\epsilon(t{-}1),2\delta(t{-}1))$-differentially private.
On Steps~5 and~7, algorithm $\RecPrefix$ interacts with its database through the choosing mechanism and using the Laplace mechanism, each of which is $(\epsilon,\delta)$-differentially private. By composition (Lemma \ref{lem:composition}), we get that $\RecPrefix$ is $(2t\epsilon,2t\delta)$-differentially private.
\end{proof}

Combining Lemma~\ref{lem:RecPrefixUtility} and Lemma~\ref{lem:RecPrefixPrivacy} we obtain Theorem~\ref{thm:range_queries_upperbound}.

\subsubsection{Informal Discussion and Open Questions}
An natural open problem is to close the gap between our (roughly) $2^{\log^* |X|}$ upper bound on the sample complexity of privately solving the interior point problem (Theorem \ref{thm:range_queries_upperbound}), and our $\log^* |X|$ lower bound (Theorem \ref{thm:range-lb}). Below we describe an idea for reducing the upper bound to $\poly(\log^* |X|)$.

In our recursive construction for the lower bound, we took $n$ elements $(x_1,\ldots,x_n)$ and generated $n+1$ elements where $y_0$ is a random element (independent of the $x_i$'s), and every $x_i$ is the length of the longest common prefix of $y_0$ and $y_i$. Therefore, a change limited to one $x_i$ affects only one $y_i$ and privacy is preserved (assuming that our future manipulations on $(y_0,\ldots,y_n)$ preserve privacy). While the representation length of domain elements grows exponentially on every step, the database size grows by 1. This resulted in the $\Omega(\log^*|X|)$ lower bound.

In $\RecPrefix$ on the other hand, every level of recursion shrank the database size by a factor of $\frac{1}{2}$, and hence, we required a sample of (roughly) $2^{\log^*|X|}$ elements. Specifically, in each level of recursion, two input elements $y_{2j-1},y_{2j}$ were paired and a new element $z_j$ was defined as the length of their longest common prefix. As with the lower bound, we wanted to ensure that a change limited to one of the inputs affects only one new element, and hence, every input element is paired only once, and the database size shrinks.

If we could pair input elements {\em twice} then the database size would only be reduced additively (which will hopefully result in a $\poly(\log^*|X|)$ upper bound). However, this must be done carefully, as we are at risk of deteriorating the privacy parameter $\epsilon$ by a factor of $2$ and thus remaining with an exponential dependency in $\log^*|X|$.
Consider the following thought experiment for pairing elements.
\begin{center}
\noindent\fbox{
\parbox{0.95\textwidth}{
{\bf Input:} $(x_1,\ldots,x_n)\in X^n$.
\begin{enumerate}[topsep=-1pt, rightmargin=10pt]
\item Let $(y_1^0,\ldots,y_n^0)$ denote a random permutation of $(x_1,\ldots,x_n)$.
\item For $t=1$ to $\log^*|X|$:
\begin{enumerate}[label={}, topsep=-1pt, rightmargin=10pt]
	\item For $i=1$ to $(n{-}t)$, let $y_i^t$ be the length of the longest common prefix of $y_i^{t-1}$ and $y_{i+1}^{t-1}$.
\end{enumerate}
\end{enumerate}
}}
\end{center}

As (most of the) elements are paired twice on every step, the database size reduces additively. In addition, every input element $x_i$ affects at most $t+1$ elements at depth $t$, and the privacy loss is acceptable. However, this still does not solve the problem. Recall that every iteration of $\RecPrefix$ begins by randomly shuffling the inputs. Specifically, we needed to ensure that (w.h.p.) the number of ``close'' pairs is limited. The reason was that if a ``not close'' pair agrees on a prefix $L$, then $L$ is the prefix ``a lot'' of other elements as well, and we could privately identify $L$.
In the above process we randomly shuffled only the elements at depth $0$. Thus we do not know if the number of ``close'' pairs is small at depth $t>0$. On the other hand, if we changed the pairing procedure to shuffle at every step, then each input element $x_i$ might affect $2^t$ elements at depth $t$, causing the privacy loss to deteriorate rapidly.
}

\stoctext{
\section{Equivalences with the Interior Point Problem}
We show that under $(\epsilon,\delta)$-differential privacy the interior point problem is equivalent with each of the following three problems: (1) Query release for threshold functions, (2) Distribution learning with respect to Kolmogorov distance, and (3) Proper PAC learning of threshold functions.
Hence, our bounds from Section~\ref{sec:lowerBound} translate to new bounds on the sample complexity of those three problems.
Here we only state the equivalences; see the full version of the paper for more details.
}

\stocrm{
\section{Query Release and Distribution Learning}

\subsection{Definitions}

Recall that a {\em counting query} $q$ is a predicate $q: X \to \{0, 1\}$. For a database $D = (x_1, \dots, x_n) \in X^n$, we write $q(D)$ to denote the average value of $q$ over the rows of $D$, i.e.
$q(D) = \frac{1}{n} \sum_{i = 1}^n q(x_i).$
In the query release problem, we seek differentially private algorithms that can output approximate answers to a family of counting queries $Q$ simultaneously.

\begin{definition}[Query Release]
Let $Q$ be a collection of counting queries on a data universe $X$, and let $\alpha, \beta > 0$ be parameters. For a database $D \in X^n$, a sequence of answers $\{a_q\}_{q \in Q}$ is \emph{$\alpha$-accurate} for $Q$ if $|a_q - q(D)| \le \alpha$ for every $q\in Q$. An algorithm $A: X^n \to \R^{|Q|}$ is \emph{$(\alpha, \beta)$-accurate for $Q$} if for every $D \in X^n$, the output $A(D)$ is $\alpha$-accurate for $Q$ with probability at least $1-\beta$ over the coins of $A$. The sample complexity of the algorithm $A$ is the database size $n$.
\end{definition}

We are interested in the query release problem for the class $\thresh_X$ of \emph{threshold queries}, which we view as a class of counting queries.

We are also interested in the following \emph{distribution learning} problem, which is very closely related to the query release problem.

\begin{definition}[Distribution Learning with respect to $Q$]
Let $Q$ be a collection of counting queries on a data universe $X$. Algorithm $A$ is an $(\alpha, \beta)$-accurate \emph{distribution learner with respect to $Q$} with sample complexity $n$ if for all distributions $\cD$ on $X$, given an input of $n$ samples $D = (x_1, \dots, x_n)$ where each $x_i$ is drawn i.i.d.\ from $\cD$, algorithm $A$ outputs a distribution $\cD'$ on $X$ (specified by its PMF) satisfying $d_{Q}(\cD, \cD') \triangleq \sup_{q \in Q}|\E_{x \sim \cD}[q(x)]- \E_{x \sim \cD'}[q(x)]| \le \alpha$ with probability at least $1 - \beta$.
\end{definition}

We highlight two important special cases of the distance measure $d_Q$ in the distribution learning problem. First, when $Q$ is the collection of \emph{all} counting queries on a domain $X$, the distance $d_{Q}$ is the \emph{total variation distance} between distributions, defined by
\[d_{\mathrm{TV}}(\cD, \cD') \triangleq \sup_{S \subseteq X} |\Pr_{x \sim \cD}[x \in S]- \Pr_{x \sim \cD'}[x \in S]|.\]
 Second, when $X$ is a totally ordered domain and $Q = \thresh_X$, the distance $d_Q$ is the \emph{Kolmogorov }(or CDF) distance. A distribution learner in the latter case may as well output a CDF that approximates the target CDF in $\ell_\infty$ norm. Specifically, we define

\begin{definition}[Cumulative Distribution Function (CDF)]
Let $\cD$ be a distribution over a totally ordered domain $X$. The CDF $F_{\cD}$ of $\cD$ is defined by $F_{\cD}(t) = \Pr_{x \sim \cD}[x \le t]$. If $X$ is finite, then any function $F: X \to [0, 1]$ that is non-decreasing with $F(\max X) = 1$ is a CDF.
\end{definition}

\begin{definition}[Distribution Learning with respect to Kolmogorov distance]
Algorithm $A$ is an $(\alpha, \beta)$-accurate \emph{distribution learner with respect to Kolmogorov distance} with sample complexity $n$ if for all distributions $\cD$ on a totally ordered domain $X$, given an input of $n$ samples $D = (x_1, \dots, x_n)$ where each $x_i$ is drawn i.i.d.\ from $\cD$, algorithm $A$ outputs a CDF $F$ with $\sup_{x \in X} |F(x) - F_{\cD}(x)|$ with probability at least $1 - \beta$.
\end{definition}

The query release problem for a collection of counting queries $Q$ is very closely related to the distribution learning problem with respect to $Q$. In particular, solving the query release problem on a dataset $D$ amounts to learning the empirical distribution of $D$. Conversely, results in statistical learning theory show that one can solve the distribution learning problem by first solving the query release problem on a sufficiently large random sample, and then fitting a distribution to approximately agree with the released answers.
\stocrm{The requisite size of this sample (without privacy considerations) is characterized by a combinatorial measure of the class $Q$ called the VC dimension:
\begin{definition}
Let $Q$ be a collection of queries over domain $X$. A set $S = \{x_1, \dots, x_k\} \subseteq X$ is \emph{shattered} by $Q$ if for every $T \subseteq [k]$ there exists $q \in Q$ such that $T=\{i:q(x_i) = 1\}$. The \emph{Vapnik-Chervonenkis (VC) dimension} of $Q$, denoted $\VC(Q)$, is the cardinality of the largest set $S \subseteq X$ that is shattered by $Q$.
\end{definition}

It is known \cite{AnthonyBa09} that solving the query release problem on $256 \VC(Q) \ln(48/\alpha\beta) / \alpha^2$ random samples yields an $(\alpha, \beta)$-accurate distribution learner for a query class $Q$.
}

\subsection{Equivalences with the Interior Point Problem}
\subsubsection{Private Release of Thresholds vs. the Interior Point Problem}
}
\stoctext{\paragraph{Private Release of Thresholds vs. the Interior Point Problem.}}
We show that the problems of privately releasing thresholds and solving the interior point problem are equivalent.

\begin{theorem} \label{thm:sanitization-vs-range}
Let $X$ be a totally ordered domain. Then,
\begin{enumerate}
	\item If there exists an $(\eps,\delta)$-differentially private algorithm that is able to release threshold queries on $X$ with $(\alpha,\beta)$-accuracy and sample complexity $n/(8\alpha)$, then there is an $(\eps, \delta)$-differentially private algorithm that solves the interior point problem on $X$ with error $\beta$ and sample complexity $n$.
	
	\item If there exists a $(1, \delta)$-differentially private algorithm solving the interior point problem on $X$ with error $O(\alpha\beta)$ and sample complexity $m$, then there is an $(\eps, \delta)$-differentially private algorithm for releasing threshold queries with $(\alpha, \beta)$-accuracy and sample complexity
$$
n= O\left(
\frac{m}{\alpha\eps}
+
\frac{\log\left(1/\delta\right)}{\alpha\eps}
+
\frac{\log\left(1/\beta\right)\log^{2.5}\left(1/\alpha\right)}{\alpha\eps}
\right).
$$
\end{enumerate}

\end{theorem}

For the first direction, observe that an algorithm for releasing thresholds could easily be used for solving the interior point problem. \stocrm{Formally,

\begin{proof}[Proof of Theorem~\ref{thm:sanitization-vs-range} item 1]
Suppose}\stoctext{More formally, suppose}
$\cA$ is a private $(\alpha,\beta)$-accurate algorithm for releasing thresholds over $X$ for databases of size $\frac{n}{8\alpha}$. Define $\cA'$ on databases of size $n$ to pad the database with an equal number of $\min\{X\}$ and $\max\{X\}$ entries, and run $\cA$ on the result. We can now return any point $t$ for which the approximate answer to the query $c_t$ is $(\frac{1}{2}\pm\alpha)$ on the (padded) database.
\stocrm{\end{proof}}

\stocrm{We now show the converse, i.e., that the problem of releasing thresholds can be reduced to the interior point problem. Specifically, we reduce the problem to a combination of solving the interior point problem, and of releasing thresholds on a much smaller data universe.
The latter task is handled by the following algorithm.

\begin{lemma}[\cite{DworkNaPiRo10}] \label{lem:prefix-tree}
For every finite data universe $X$, and $n\in\N$, $\eps,\beta>0$,
there is an $\eps$-differentially private algorithm $A$ that releases all threshold queries on $X$ with $(\alpha,\beta)$-accuracy for
\[\alpha=\frac{4 \log(1/\beta)\log^{2.5}|X|}{\eps n}.\]
\end{lemma}

}

\stoctext{For proving the second direction of the equivalence, we reduce the problem of releasing thresholds to a combination of solving the interior point problem, and of releasing thresholds on a much smaller data universe.}
At a high level, the reduction we present consists of two steps: (1) A ``partitioning procedure'' that privately identifies $\approx 1/\alpha$ representatives that partition the data into blocks
of size roughly $\alpha n$, and (2) answering threshold queries on just the set of representatives. 
From this we can well-approximate all threshold queries. Moreover, since there are only $O(1/\alpha)$ representatives,
\stocrm{the base algorithm, mentioned above, gives only $\polylog(1/\alpha)$ error for these answers.}
\stoctext{we can use the results of~\cite{DworkNaPiRo10} in order to incur only $\polylog(1/\alpha)$ error for these answers.}

\begin{remark}
In an earlier version of this paper, the proof of Item~2 of Theorem~\ref{thm:sanitization-vs-range} had an error in the privacy analysis, which was pointed out to us by Haim Kaplan in 2021 and Roodabeh Safavi in 2024. We repair it here by showing that our original algorithm is $\left(O(1),\delta\right)$-differentially private (rather than $\left(\eps,\delta\right)$-differentially private as we claimed) and then applying privacy amplification via subsampling.

Subsequent work introduced several alternative ``partitioning procedures'' that could be used in our reduction in place of our original partitioning procedure. These include the partitioning procedures proposed in \cite{KaplanSS22} and \cite{Cohen0NSS23}. Additionally, in a personal communication, \cite{JalajPersonalComm} suggested another partitioning procedure, similar to that of \cite{KaplanSS22}.
\end{remark}

\stocrm{

\begin{proof}[Proof of Theorem~\ref{thm:sanitization-vs-range} item 2]
Let $R: X^* \to X$ be a $(1, \delta)$-differentially private algorithm solving the interior point problem on $X$ with error $O(\alpha\beta)$ and sample complexity $m$.
We may actually assume that $R$ is differentially private in the sense that if $D \in X^*$ and $D'$ differs from $D$ up to the addition or removal of a row, then for every $S \subseteq X$, $\Pr[R(D) \in S] \le e\cdot \Pr[R(D') \in S] + \delta$, and that $R$ solves the interior point problem with probability at least $1-O(\alpha\beta)$ whenever its input is of size at least $m$. This is because we can pad databases of size less than $m$ with an arbitrary fixed element, and subsample the first $m$ entries from any database with size greater than $m$.

Consider the following algorithm for answering thresholds on databases $D \in X^n$ for $n > m$:
\begin{algorithm}[H]
\caption{$Thresh(D)$}
\textbf{Input:} Database $D \in X^n$.

\smallskip
\textbf{Tools used:} A $(1,\delta)$-differentially private algorithm $R$ for solving the interior point problem on $X$ with error $\alpha\beta/56000$ and sample complexity $m$. A differentially private algorithm $A$ for releasing all threshold queries on $X$, as in Lemma~\ref{lem:prefix-tree}. 

\begin{enumerate}

\item Let $S$ be a subsample of $D$, where every point in $D$ is sampled into $S$ with probability $p=\frac{\eps}{8e^4}$. Denote $|S|=\hat{n}$. 

\item Sort $S$ in nondecreasing order $x_1 \le x_2 \le \dots \le x_{\hat{n}}$.

\item Denote $\hat{m}=\max\{m\;,\, \frac{\alpha\eps n}{7000} \}$, and set $k = \eps n/\hat{m}$. Let $t_0 = 1,\; t_1 = t_0 + \nu_1,\; t_2 = t_1 + \nu_2 \dots,\; t_k = t_{k-1} + \nu_k$ where each $\nu_\ell \sim \lceil \hat{m}+4+\log(\frac{k}{\beta\delta})+\Lap(1)\rceil$ independently.

\item Divide $S$ into blocks $S_1, \dots, S_k$, where $S_\ell = (x_{t_{\ell-1}}, \dots, x_{t_{\ell}-1})$ (setting $x_j = \max X$ if $j > n$; note some $S_\ell$ may be empty).

\item Let $r_0 = \min X$, $r_1 = R(S_1), \dots, r_k = R(S_k)$.

\item Define $\hat{D}$ from $D$ by replacing each $x_j$ with the largest $r_\ell$ for which $r_\ell \le x_j$.

\item Run algorithm $A$ from Lemma~\ref{lem:prefix-tree} on $\hat{D}$ over the universe $\{r_0, r_1, \dots, r_k\}$ to obtain threshold query answers $a_{r_0}, a_{r_1}, \dots, a_{r_k}$. Use privacy parameter $\eps/2$ and confidence parameter $\beta/2$.

\item Answer arbitrary threshold queries by interpolation, i.e. for $r_\ell \le t < r_{\ell+1}$, set $a_t = a_{r_\ell}$.

\item Output $(a_t)_{t\in X}$.
\end{enumerate}
\end{algorithm}

\paragraph{Privacy}
Let $\BBB_{2-5}$ denote the algorithm consisting of Steps~2-5 of Algorithm $Thresh$. The input of $\BBB_{2-5}$ is a database $S$ of size $\hat{n}$, and its output are the values $r_0,r_1,\dots,r_k$. We first show that $\BBB_{2-5}$ is $(O(1),O(\delta))$-differentially private. To this end, let $S=(x_1,\dots,x_{\hat{n}})$ where $x_1\leq x_2\leq\dots x_{\hat{n}}$, and consider a neighboring database $S'=(x_1,\dots,x'_i,\dots,x_{\hat{n}})$. Assume without loss of generality that $x'_i\geq x_{i+1}$, and suppose
\[x_1 \le \dots \le x_{i-1} \le x_{i+1} \le \dots \le x_{j} \le x_{i}' \le x_{j+1} \le \dots \le x_{n'}.\]

We define a mapping $\pi:\R^k \rightarrow \R^k$ from noise vectors $\nu = (\nu_1, \dots, \nu_k)$ during the execution on $S$ to noise vectors $\nu' = (\nu'_1, \dots, \nu'_k)$ during the execution
on $S'$ such that $S$ partitioned according to $\nu$ and $S'$ partitioned according to $\nu'=\pi(\nu)$ differ on at most two blocks.
Specifically, if $\ell, r$ are the indices for which $t_{\ell-1} \le i < t_{\ell}$ and $t_{r - 1} \le j < t_r$ (we may have $\ell=r$), then we can take $\nu'_{\ell} = \nu_\ell - 1$ and $\nu'_{r} = \nu_r + 1$ with $\nu' = \nu$ at every other index.
Note that $S$ partitioned into $(S_1,\ldots,S_k)$ according to $\nu$ differs from $S'$ partitioned into $(S'_1,\ldots,S'_k)$ according to $\nu'$
by a removal of an element from one block (namely $S_\ell$) and the addition of an element to another block (namely $S_r$).

We now claim that the mapping $\pi$ is ``essentially'' a 2-to-1 mapping. Specifically, we show that every vector $\nu'$ such that $\forall t\in[k]\;\nu'_t\geq3$ could have at most 2 preimages. Indeed, let $\nu'$ be such a vector with a preimage $\nu$. Observe that all the coordinates of $\nu$ must be at least 2. In particular, the partition of $S$ according to $\nu$ and the partition of $S'$ according to $\nu'$ are both well-defined, containing no empty blocks. 

Recall that the preimage $\nu$ is obtained from $\nu'$ by increasing one coordinate and decreasing another (by 1). We will show that there is exactly one option for the index we need to decrease, and there could be at most 2 options for which index we need to increase. To this end,
let $S'_r$ be the unique block containing $x_j$ in the partition of $S'$ according to $\nu'$. 
Observe that $x_j$ cannot be the last element in this block, as otherwise $\nu'$ would not have any preimage (the mapping $\pi$ does not generate partitions containing a block that ends with $x_j$). Thus, we must set $\nu_r=\nu'_r-1$ in order for $\nu$ to be a preimage of $\nu'$.

Next let $\ell_1,\ell_2$ denote the indices of the blocks containing $x_{i-1}$ and $x_{i+1}$ in the partition of $S'$ according to $\nu'$ (it could be that $\ell_1=\ell_2$, or that $\ell_2=\ell_1+1$). Note that in the partition of $S$ according to $\nu$, it must be that $x_i$ belongs to the same block as at least one of $x_{i-1},x_{i+1}$ (since all blocks are of size at least 2). Thus, there could be at most two options for the index of the coordinate in $\nu$ that is increased by 1 compared to $\nu'$: Either $\nu_{\ell_1}=\nu'_{\ell_1}+1$ or $\nu_{\ell_2}=\nu'_{\ell_2}+1$. All other coordinates are the same. This shows that there could be at most 2 options for the preimage $\nu$.

Now let $F$ be a set of possible outcomes of Algorithm $\BBB_{2-5}$. We have that
\begin{align*}
\Pr[\BBB_{2-5}(S)\in F] &= \sum_{\nu} \Pr[\nu]\cdot \Pr[\BBB_{2-5}(S)\in F | \nu]\\
&\leq \Pr_{\nu}[\min\{v_t\}<4]
+
\sum_{\nu\,:\,\min\{v_t\}\geq4} \Pr[\nu]\cdot \Pr[\BBB_{2-5}(S)\in F | \nu]\\
&\leq \delta
+
\sum_{\nu\,:\,\min\{\nu_t\}\geq4} \Pr[\nu]\cdot \Pr[\BBB_{2-5}(S)\in F | \nu]\\
&\leq \delta
+
\sum_{\nu\,:\,\min\{\nu_t\}\geq4} e^2\cdot\Pr[\pi(\nu)]\cdot \Pr[\BBB_{2-5}(S)\in F | \nu]\\
&\leq \delta + \sum_{\nu\,:\,\min\{\nu_t\}\geq4} e^2\cdot \Pr[\pi(\nu)]\cdot \left(e^2\cdot\Pr[\BBB_{2-5}(S')\in F | \pi(\nu)]+2\delta\right)\\
&\leq \delta + 2\cdot\sum_{\nu'\in{\rm Range}(\pi)} e^2\cdot \Pr[\nu']\cdot \left(e^2\cdot\Pr[\BBB_{2-5}(S')\in F | \nu']+2\delta\right)\\
&\leq \delta + 2\cdot\sum_{\nu} e^2\cdot \Pr[\nu]\cdot \left(e^2\cdot\Pr[\BBB_{2-5}(S')\in F | \nu]+2\delta\right)\\
&= e^{4+\ln(2)}\cdot \Pr[\BBB_{2-5}(S')\in F]  + (4e^2+1)\delta,
\end{align*}
where the second inequality follows by a union bound on the Laplace noises sampled in Step 3, 
the third inequality follows 
since the noise vector $\nu$ is sampled with density at most $e^2$ times the density of $\pi(\nu)$, the forth inequality follows from the privacy guarantees of $R$ (recall that at most two applications of $R$ are affected), and the fifth inequality is because every $\nu'\in{\rm Range}(\pi)$ could have at most 2 preimages $\nu$ satisfying $\min_t\{\nu_t\}\geq4$. This shows that Algorithm $\BBB_{2-5}$ is $(O(1),O(\delta))$-differentially private.

Now consider algorithm $\BBB_{1-5}$, consisting of Steps 1-5 of Algorithm $Thresh$. The input of $\BBB_{1-5}$ is a database $D$ of size $n$ and its output are the values $r_0,r_1,\dots,r_k$. Note that $\BBB_{1-5}$ can be viewed as running $\BBB_{2-5}$ on a subsample. We leverage the following result for showing that $\BBB_{1-5}$ is differentially private with boosted privacy parameters:

\begin{theorem}[\cite{KasiviswanathanLeNiRaSm07}]
Let $\AAA$ be an $(\eps_0,\delta)$-differentially private algorithm. Fix $0\leq p\leq0.5$ and let $\hat{\AAA}$ be the algorithm that takes a database $D$, construct a dataset $\hat{D}$ by sampling every point from $D$ independently with probability $p$ (a.k.a.\ Poisson sampling), and runs $\AAA$
on $\hat{D}$. Then, $\hat{\AAA}$ is
$(2p(e^{\eps_0}-1),p\delta)$-differentially private
\end{theorem}

Using this theorem, we get that $\BBB_{1-5}$ satisfies $(\frac{\eps}{2},\delta)$-differential privacy. Finally, note that algorithm $Thresh$ can be viewed as the composition of $\BBB_{1-5}$ with the algorithm $A$ from Lemma~\ref{lem:prefix-tree}, which is $(\frac{\eps}{2},0)$-differentially private. Algorithm $Thresh$ is, therefore, $(\eps,\delta)$-differentially private by composition.


\paragraph{Utility} We can produce accurate answers to every threshold function as long as
\begin{enumerate}
    \item For every consecutive collection of $t=\alpha n/2$ points from $D$, call them $x_{j_1}\leq x_{j_2}\leq\dots\leq x_{j_t}$, there exist a value $r_j$ (computed in Step 5) satisfying $x_{j_1}\leq r_j\leq x_{j_t}$.

    \item The answers obtained from executing the algorithm from Lemma~\ref{lem:prefix-tree} all have additive (unnormalized) error at most $\alpha n/2$.

\end{enumerate}

Item 2 holds with probability at least $1-\beta/2$ by Lemma~\ref{lem:prefix-tree}, provided that
$n\geq \Omega\left(\frac{1}{\alpha\eps} \log(\frac{1}{\beta})\log^{2.5}(\frac{1}{\alpha})\right)$. 
As for Item 1, we show that it holds whenever the following events occur:

\begin{enumerate}
    \item[(a)] For every consecutive collection of $t=\alpha n/2$ points from $D$, at least $\frac{pt}{4}=\frac{\alpha\eps n}{64e^4}$ of them are sampled into $S$.

    \item[(b)] Every database $S_i$ has size at least $\hat{m}$ and at most $\hat{m}+4+2\log(\frac{k}{\beta\delta})$.

    \item[(c)]  The partitioning exhausts the database $S$, i.e. every element of $S$ is in some $S_i$.

    \item[(d)] Every execution of $R$ succeeds at finding an interior point,
\end{enumerate}

Indeed, suppose that (a)-(d) occur. Then, by (a), every consecutive collection of $t$ points in $D$ has at least $\frac{\alpha\eps n}{64e^4}$ consecutive representatives in $S$, which is more than twice times the maximal possible block size (as stated by event (b)), provided that $n\geq\Omega(\frac{m}{\alpha\eps}+\frac{1}{\alpha\eps}\log(\frac{1}{\alpha\beta\delta}))$. Thus, by (b) and (c), there must be a block $S_i$ in the partition of $S$ that is completely contained in this collection of representatives. Thus, by (d), at least one $r_j$ must be an interior point of these $t$ points in $D$.

Now note that if $t\geq\Omega(\frac{1}{p}\ln(\frac{1}{\alpha\beta}))$, then Item (a) happens with probability at least $1-\beta/8$ by the Chernoff bound\footnote{This calculation includes a union bound over $k$ disjoint (consecutive) sequences of $\frac{t}{2}$ elements in $D$. It suffices to union bound over these $O(\frac{1}{\alpha})$ sequences because every sequence of length $t$ must contain at least one such sequence of length $t/2$.}. Also note that Item (b) happens with probability at least $1-\beta/8$ by standard tails bounds for the Laplace distribution. Next, when (b) holds, then Item (c) occurs whenever $|S|\leq\eps n$, which happens with probability at least $1-\beta/8$ by the Chernoff bound, provided that $n\geq \Omega\left(\frac{1}{\eps}\ln(\frac{1}{\beta})\right)$. Finally, when (b) holds, then Item (d) holds with probability at least $1-\beta/8$ by the guarantees of algorithm $R$ (recall that we required the failure probability of $R$ to be $\ll \alpha\beta\leq\frac{\beta}{k}$).

Overall, with probability at least $1-\beta$, our answers are $\alpha$-accurate, provided that
$$
n\geq \Omega\left(
\frac{m}{\alpha\eps}
+
\frac{1}{\alpha\eps}\log\left(\frac{1}{\delta}\right)
+
\frac{1}{\alpha\eps} \log\left(\frac{1}{\beta}\right)\log^{2.5}\left(\frac{1}{\alpha}\right)
\right).
$$
\end{proof}

\subsubsection{Releasing Thresholds vs.\ Distribution Learning}
}

\stoctext{\paragraph{Releasing Thresholds vs.\ Distribution Learning.}}
Query release and distribution learning are very similar tasks: A distribution learner can be viewed as an algorithm for query release with small error w.r.t.\ the underlying distribution (rather than the fixed input database). We show that the two tasks are equivalent under differential privacy.

\begin{theorem}\label{thm:distlearning-release-reduction}
Let $Q$ be a collection of counting queries over a domain $X$.
\begin{enumerate}
	\item If there exists an $(\eps,\delta)$-differentially private algorithm for releasing $Q$ with $(\alpha,\beta)$-accuracy and sample complexity $n\geq 256 \VC(Q) \ln(48/\alpha\beta) / \alpha^2$, then there is an $(\eps,\delta)$-differentially private $(3\alpha, 2\beta)$-accurate distribution learner w.r.t.\ $Q$ with sample complexity $n$.

	\item If there exists an $(\eps, \delta)$-differentially private $(\alpha, \beta)$-accurate distribution learner w.r.t.\ $Q$ with sample complexity $n$, then there is an $(\eps,\delta)$-differentially private query release algorithm for $Q$ with $(\alpha,\beta)$-accuracy and sample complexity $9n$.
\end{enumerate}

\end{theorem}

The first direction follows from a standard generalization bound, showing that if a given database $D$ contains (enough) i.i.d.\ samples from a distribution $\cD$, then (w.h.p.) accuracy with respect to $D$ implies accuracy with respect to $\cD$. We remark that the sample complexity lower bound on $n$ required to apply item 1 of Theorem~\ref{thm:distlearning-release-reduction} does not substantially restrict its applicability: It is known that an $(\eps, \delta)$-differentially private algorithm for releasing $Q$ always requires sample complexity $\Omega(\VC(Q)/\alpha\eps)$ anyway \cite{BlumLiRo08}.

\stocrm{
\begin{proof}[Proof of Theorem~\ref{thm:distlearning-release-reduction}, item 1]
Suppose $\tilde{\cA}$ is an $(\eps,\delta)$-differentially private algorithm for releasing $Q$ with $(\alpha,\beta)$-accuracy and sample complexity $n\geq 256 \VC(Q) \ln(48/\alpha\beta) / \alpha^2$.
Fix a distribution $\cD$ over $X$ and consider a database $D$ containing $n$ i.i.d.\ samples from $\cD$.
Define the algorithm $\cA$ that on input $D$ runs $\tilde{\cA}$ on $D$ to obtain answers $a_q$ for every query $q \in Q$. Afterwards, algorithm $\cA$ uses linear programming \cite{DworkNaReRoVa09} to construct a distribution $\cD'$ that such that $|a_q - q(\cD')| \le \alpha$ for every $q \in Q$, where $q(\cD') \triangleq \E_{x \sim \cD'}[q(x)]$. This reconstruction always succeeds as long as the answers $\{a_q\}$ are $\alpha$-accurate, since the empirical distribution of $D$ is a feasible point for the linear program. Note that $\cA$ is $(\eps,\delta)$-differentially private (since it is obtained by post-processing $\tilde{\cA}$).

We first argue that $q(\cD')$ is close to $q(D)$ for every $q\in Q$, and then argue that $q(D)$ is close to $q(\cD)$. By the utility properties of $\tilde{\cA}$, with all but probability $\beta$,
$$
|q(\cD') - q(D)| \le |q(\cD') - a_q| + |a_q - q(D)| \le 2\alpha.
$$
for every $q \in Q$.

We now use the following generalization theorem to show that (w.h.p.) $q(D)$ is close to $q(\cD)$ for every $q \in Q$.

\begin{theorem}[\cite{AnthonyBa09}] \label{thm:dist-generalization}
Let $Q$ be a collection of counting queries over a domain $X$. Let $D = (x_1, \dots, x_n)$ consist of i.i.d.\ samples from a distribution $\cD$ over $X$. If $d = \VC(Q)$, then
\[\Pr\left[\sup_{q \in Q} |q(D) - q(\cD)| > \alpha\right] \le 4 \left(\frac{2en}{d}\right)^d\exp\left(-\frac{\alpha^2n}{8}\right).\]
\end{theorem}


Using the above theorem, together with the fact that $n \ge 256 \VC(Q) \ln(48/\alpha\beta) / \alpha^2$, we see that except with probability at least $1 - \beta$ we have that $|q(D) - q(\cD))|\leq \alpha$ for every $q \in Q$. By a union bound (and the triangle inequality) we get that $\cA$ is $(3\alpha,2\beta)$-accurate.
\end{proof}

In the special case where $Q = \thresh_X$ for a totally ordered domain $X$, corresponding to distribution learning under Kolmogorov distance, the above theorem holds as long as $n \ge 2\ln(2/\beta)/\alpha^2$.
This follows from using the Dvoretzky-Kiefer-Wolfowitz inequality \cite{DvoretzkyKiWo56,massart1990} in place of Theorem \ref{thm:dist-generalization}.

\begin{theorem}
If there exists an $(\eps,\delta)$-differentially private algorithm for releasing $\thresh_X$ over a totally ordered domain $X$ with $(\alpha,\beta)$-accuracy and sample complexity $n\geq 2\ln(2/\beta)/\alpha^2$, then there is an $(\eps,\delta)$-differentially private $(2\alpha, 2\beta)$-accurate distribution learner under Kolmogorov distance with sample complexity $n$.
\end{theorem}


We now show the other direction of the equivalence.

\begin{lemma}\label{lem:nonprivateSubsample}
Suppose $\cA$ is an $(\eps, \delta)$-differentially private and $(\alpha, \beta)$-accurate distribution learner w.r.t.\ a concept class $Q$ with sample complexity $n$. Then there is an $(\eps, \delta)$-differentially private algorithm $\tilde{\cA}$ for releasing $Q$ with $(\alpha,\beta)$-accuracy and sample complexity $9n$.
\end{lemma}
}

\stoctext{For the second direction of the equivalence,}
\stocrm{To construct the algorithm $\tilde{\cA}$,}
we note that a distribution learner must perform well on the uniform distribution on the rows of any fixed database, and thus must be useful for releasing accurate answers for queries on such a database. Thus if we have a distribution learner $\cA$, the mechanism $\tilde{\cA}$ that samples $m$ rows (with replacement) from its input database $D \in (X \times \{0, 1\})^n$ and runs $\cA$ on the result should output accurate answers for queries with respect to $D$. The random sampling has two competing effects on privacy. On one hand, the possibility that an individual is sampled multiple times incurs additional privacy loss. On the other hand, if $n > m$, then a ``secrecy-of-the-sample'' argument shows that random sampling actually improves privacy, since any individual is unlikely to have their data affect the computation at all. We show that if $n$ is only a constant factor larger than $m$, these two effects offset, and the resulting mechanism is still differentially private.

\stocrm{
\begin{proof}
Consider a database $D \in X^{9n}$. Let $\cD$ denote the uniform distribution over the rows of $D$, and let $\cD'$ be the distribution learned. Consider the algorithm $\tilde{\cA}$ that subsamples (with replacement) $n$ rows from $D$ and runs $\cA$ on it to obtain a distribution $\cD'$. Afterwards, algorithm $\tilde{\cA}$ answers every threshold query $q \in Q$ with $a_q = q(\cD') \triangleq \E_{x \sim \cD'}[q(x)]$.

Note that drawing $n$ i.i.d.\ samples from $\cD$ is equivalent to subsampling $n$ rows of $D$ (with replacement).
Then with probability at least $1 - \beta$, the distribution $\cD'$ returned by $\cA$ is such that for every $x\in X$
$$
|q(\cD')-q(D)| = |q(\cD')-q(\cD)| \leq \alpha,
$$
showing that $\tilde{\cA}$ is $(\alpha,\beta)$-accurate.
\end{proof}

We'll now use a secrecy-of-the-sample argument (refining an argument that appeared implicitly in \cite{KasiviswanathanLeNiRaSm07}), to show that $\tilde{\cA}$ (from Lemma~\ref{lem:nonprivateSubsample}) is differentially private whenever $\cA$ is differentially private.

\begin{lemma}\label{lem:secrecy-of-the-sample}
Fix $\epsilon\leq1$ and let $\cA$ be an $(\eps, \delta)$-differentially private algorithm operating on databases of size $m$.
For $n\geq2m$, construct an algorithm $\tilde{\cA}$ that on input a database $D$ of size $n$ subsamples (with replacement) $m$ rows from $D$ and runs $\cA$ on the result. Then $\tilde{\cA}$ is $( \tilde{\eps} , \tilde{\delta} )$-differentially private for
$$\tilde{\eps}=6\eps m /n \quad \text{and} \quad \tilde{\delta}= \exp(6\eps m/n)\frac{4m}{n}\cdot\delta.$$
\end{lemma}

\begin{proof}
Let $D, D'$ be adjacent databases of size $n$, and suppose without loss of generality that they differ on the last row. Let $T$ be a random variable denoting the multiset of indices sampled by $\tilde{\cA}$, and let $\ell(T)$ be the multiplicity of index $n$ in $T$. Fix a subset $S$ of the range of $\tilde{\cA}$. For each $k = 0, 1, \dots, m$ let
\begin{align*}
p_k &= \Pr[\ell(T) = k] = {m \choose k} n^{-k} (1 - 1/n)^{m - k} = {m \choose k} (n-1)^{-k} (1 - 1/n)^m ,\\
q_k &= \Pr[\cA(D|_T) \in S | \ell(T) = k], \\
q'_k &= \Pr[\cA(D'|_T) \in S | \ell(T) = k].
\end{align*}
Here, $D|_T$ denotes the subsample of $D$ consisting of the indices in $T$, and similarly for $D'|_T$. Note that $q_0 = q'_0$, since $D|_T = D'|_T$ if index $n$ is not sampled. Our goal is to show that
\[\Pr[\tilde{\cA}(D) \in S] = \sum_{k=0}^{m} p_k q_k \le e^{\tilde{\eps}} \sum_{k=0}^{m} p_kq'_k + \tilde{\delta} = e^{\tilde{\eps}} \Pr[\tilde{\cA}(D') \in S] + \tilde{\delta}.\]

To do this, observe that by privacy,
$q_k \le e^{\eps}q_{k-1} + \delta$, so
\[q_k \le e^{k\eps} q_0 + \frac{e^{k\eps} - 1}{e^{\eps} - 1}\delta.\]

Hence,
\begin{eqnarray}
\Pr[\tilde{\cA}(D) \in S] &=& \sum_{k=0}^{m} p_k q_k \nonumber \\
&\leq& \sum_{k=0}^{m}  {m \choose k} (n-1)^{-k} (1 - 1/n)^m \left(  e^{k\eps} q_0 + \frac{e^{k\eps} - 1}{e^{\eps} - 1}\delta \right) \nonumber \\
&=& q_0 (1 - 1/n)^m \sum_{k=0}^{m} {m \choose k} \left(\frac{e^\eps}{n-1}\right)^k
    + \frac{\delta}{e^{\eps} - 1} (1 - 1/n)^m \sum_{k=0}^{m} {m \choose k} \left(\frac{e^\eps}{n-1}\right)^k
    - \frac{\delta}{e^{\eps} - 1} \nonumber \\
&=& q_0 (1 - 1/n)^m \left( 1 + \frac{e^\eps}{n-1} \right)^m
    + \frac{\delta}{e^{\eps} - 1} (1 - 1/n)^m \left( 1 + \frac{e^\eps}{n-1} \right)^m
    - \frac{\delta}{e^{\eps} - 1} \nonumber \\
&=& q_0 \left( 1-\frac{1}{n}+\frac{e^\eps}{n} \right)^m + \frac{\left( 1-\frac{1}{n}+\frac{e^\eps}{n} \right)^m -1}{e^\eps -1}\delta. \label{eq:pkqk1}
\end{eqnarray}

Similarly, we also have that
\begin{eqnarray}
\Pr[\tilde{\cA}(D') \in S] \geq q_0 \left( 1-\frac{1}{n}+\frac{e^{-\eps}}{n} \right)^m - \frac{\left( 1-\frac{1}{n}+\frac{e^{-\eps}}{n} \right)^m -1}{e^{-\eps} -1}\delta.
\label{eq:pkqk2}
\end{eqnarray}

Combining inequalities~\ref{eq:pkqk1} and~\ref{eq:pkqk2} we get that
\begin{eqnarray*}
\Pr[\tilde{\cA}(D) \in S] &\leq& \left(\frac{1-\frac{1}{n}+\frac{e^\eps}{n}}{1-\frac{1}{n}+\frac{e^{-\eps}}{n}}\right)^m
\cdot\left\{ \Pr[\tilde{\cA}(D') \in S]  +\frac{1-\left(1-\frac{1}{n}+\frac{e^{-\eps}}{n}\right)^m}{1-e^{-\eps}}\delta \right\}
+\frac{\left(1-\frac{1}{n}+\frac{e^\eps}{n}\right)^m -1}{e^\eps -1}\delta,
\end{eqnarray*}
proving that $\cA'$ is $(\tilde{\eps},\tilde{\delta})$-differentially private for
\begin{eqnarray*}
\tilde{\eps} &\leq& m\cdot\ln\left( \frac{1+\frac{e^\eps -1}{n}}{1+\frac{e^{-\eps}-1}{n}} \right)
\leq\frac{6\eps m}{n}
\end{eqnarray*}
and
\begin{eqnarray*}
\tilde{\delta}&\leq& \exp(6\eps m/n)\frac{1-\left(1+\frac{e^{-\eps}-1}{n}\right)^m}{1-e^{-\eps}}\cdot\delta + \frac{\left(1+\frac{e^{\eps}-1}{n}\right)^m-1}{e^{\eps}-1}\cdot\delta\\
&\leq& \exp(6\eps m/n)\frac{1-\exp\left(2\frac{e^{-\eps} -1}{n/m}\right)}{1-e^{-\eps}}\cdot\delta + \frac{\exp\left(\frac{e^\eps-1}{n/m}\right)-1}{e^{\eps}-1}\cdot\delta\\
&\leq& \exp(6\eps m/n)\frac{2m}{n}\cdot\delta + \frac{2m}{n}\cdot\delta\\
&\leq& \exp(6\eps m/n)\frac{4m}{n}\cdot\delta.
\end{eqnarray*}

\end{proof}

\section{PAC Learning}

\subsection{Definitions}

A concept $c:X\rightarrow \{0,1\}$ is a predicate that labels {\em examples} taken from the domain $X$.  A \emph{concept class} $C$ over $X$ is a set of concepts over the domain $X$. A learner is given examples sampled from an unknown probability distribution $\cD$ over $X$ that are labeled according to an unknown {\em target concept} $c\in C$ and outputs a hypothesis $h$ that approximates the target concept with respect to the distribution $\cD$. More precisely,

\begin{definition}
The {\em generalization error} of a hypothesis $h:X\rightarrow\{0,1\}$ (with respect to a target concept $c$ and distribution $\cD$) is defined by
$\error_{\cD}(c,h)=\Pr_{x \sim \cD}[h(x)\neq c(x)].$
If $\error_{\cD}(c,h)\leq\alpha$ we say that $h$ is an {\em $\alpha$-good} hypothesis for $c$ on $\cD$.
\end{definition}

\begin{definition}[PAC Learning~\cite{Valiant84}]\label{def:PAC}
Algorithm $A$ is an {\em $(\alpha,\beta)$-accurate PAC learner} for a concept class $C$ over $X$ using hypothesis class $H$ with sample complexity $m$ if for all target concepts $c \in C$ and all distributions $\cD$ on $X$, given an input of $m$ samples $\db =((x_i, c(x_i)),\ldots,(x_m, c(x_m)))$, where each $x_i$ is drawn i.i.d.\ from $\cD$, algorithm $A$ outputs a hypothesis $h\in H$ satisfying
$\Pr[\error_{\cD}(c,h)  \leq \alpha] \geq 1-\beta.$

The probability is taken over the random choice of the examples in $\db$ and the coin tosses of the learner $A$. If $H\subseteq C$ then $A$ is called {\em proper}, otherwise, it is called {\em improper}.
\end{definition}

\begin{definition}
The \emph{empirical error} of a hypothesis $h$ on a labeled sample $S=((x_1,\ell_1), \dots, (x_m, \ell_m))$ is $\error_S(h) = \frac{1}{m} |\{i : h(x_i) \neq \ell_i\}|.$
If $\error_{S}(h) \le \alpha$ we say $h$ is \emph{$\alpha$-consistent} with $S$.
\end{definition}

Classical results in statistical learning theory show that
a sample of size $\Theta(\VC(C))$ is both necessary and sufficient for PAC learning a concept class $C$. \stocrm{That $O(\VC(C))$ samples suffice follows from a ``generalization'' argument: for any concept $c$ and distribution $\cD$, with probability at least $1-\beta$ over $m > O_{\alpha,\beta}(\VC(C))$ random labeled examples, \emph{every} concept $h \in C$ that agrees with $c$ on the examples has error at most $\alpha$ on $\cD$. Therefore, $C$ can be properly learned by finding any hypothesis $h \in C$ that agrees with the given examples.}

Recall the class of \emph{threshold functions}, which are concepts defined over a totally ordered domain $X$ by $\thresh_X = \{c_x : x \in X\}$ where $c_x(y) = 1$ iff $y \le x$\mnote{Didn't change this to $y \ge x$ (yet)}. The class of threshold functions has VC dimension $\VC(\thresh_X) = 1$, and hence can be learned with $O_{\alpha, \beta}(1)$ samples.


A private learner is a PAC learner that is differentially private. Following \cite{KasiviswanathanLeNiRaSm07}, we consider algorithms $A: (X \times \{0, 1\})^m \to H$, where $H$ is a hypothesis class, and require that
\begin{enumerate}
\item $A$ is an $(\alpha, \beta)$-accurate PAC learner for a concept class $C$ with sample complexity $m$, and
\item $A$ is $(\eps, \delta)$-differentially private.
\end{enumerate}


\stocrm{Note that while we require utility (PAC learning) to hold only when the database $D$ consists of random labeled examples from a distribution, the requirement of differential privacy applies to every pair of neighboring databases $D \sim D'$, including those that do not correspond to examples labeled by any concept.}

Recall the relationship between distribution learning and releasing thresholds, where accuracy is measured w.r.t.\ the underlying distribution in the former and w.r.t.\ the fixed input database in the later. Analogously, we now define the notion of an {\em empirical learner} which is similar to a PAC learner where accuracy is measured w.r.t.\ the fixed input database.

\begin{definition}[Empirical Learner]
Algorithm $A$ is an {\em $(\alpha,\beta)$-accurate empirical learner} for a concept class $C$ over $X$ using hypothesis class $H$ with sample complexity $m$ if for every $c\in C$ and for every database $\db =((x_i, c(x_i)),\ldots,(x_m, c(x_m)))\in(X\times\{0,1\})^m$ algorithm $A$ outputs a hypothesis $h\in H$ satisfying
$\Pr[\error_{\db}(c,h)  \leq \alpha] \geq 1-\beta.$

The probability is taken over the coin tosses of $A$.
\end{definition}

Note that without privacy (and ignoring computational efficiency) identifying a hypothesis with small empirical error is trivial for every concept class $C$ and for every database of size at least $1$. This is not the case with $(\eps,\delta)$-differential privacy,\footnote{The lower bound in Theorem~\ref{thm:empiricalVC} also holds for {\em label private} empirical learners, that are only required to provide privacy for the labels in the database.} and the sample complexity of every empirical learner for a concept class $C$ is at least $\Omega(\VC(C))$:

\begin{theorem}\label{thm:empiricalVC}
For every $\alpha,\beta\leq1/8$, every $\delta\leq\frac{1}{8n}$ and $\epsilon>0$,
if $\cA$ is an $(\eps,\delta)$-differentially private $(\alpha,\beta)$-accurate empirical learner for a class $C$ with sample complexity $n$, then
$n=\Omega\left(\frac{1}{\alpha\epsilon}\VC(C)\right)$.
\end{theorem}

The proof of Theorem~\ref{thm:empiricalVC} is very similar the analysis of~\cite{BlumLiRo08} for lower bounding the sample complexity of releasing approximated answers for queries in the class $C$.
As we will see in the next section, at least in some cases (namely, for threshold functions) the sample complexity must also have some dependency in the size of the domain $X$.

\begin{proof}[Proof of Theorem~\ref{thm:empiricalVC}]
Fix $d<\VC(C)$, let $x_0,x_1,x_2,\ldots,x_d$ be {\em shattered} by $C$, and denote $S=\{x_1,\ldots,x_d\}$.
Let $D$ denote a database containing $(1-8\alpha)n$ copies of $x_0$ and $8\alpha n/d$ copies of every $x_i\in S$.
For a concept $c$ we use $D_c$ to denote the database $D$ labeled by $c$.
We will consider concepts that label $x_0$ as 0, and label exactly half of the elements in $S$ as 1.
To that end, initiate $\tilde{C}=\emptyset$, and for every subset $S'\subseteq S$ of size $|S'|=|S|/2$, add to $\tilde{C}$ one concept $c\in C$ s.t.\ $c(x_0)=0$ and for every $x_i\in S$ it holds that $c(x_i)=1$ iff $x_i\in S'$ (such a concept exists since $S\cup\{x_0\}$ is shattered by $C$).

Now, let $c\in\tilde{C}$ be chosen uniformly at random, let $x\in S$ be a random element s.t.\ $c(x)=1$, and let $y\in S$ be a random element s.t.\ $c(y)=0$. Also let $c'\in\tilde{C}$ be s.t.\ $c'(x)=0$, $c'(y)=1$, and $c'(x_i)=c(x_i)$ for every $x_i\in S\setminus\{x,y\}$. Note that the marginal distributions on $c$ and on $c'$ are identical, and denote $h=\cA(D_c)$ and $h'=\cA(D_{c'})$.

Observe that $x$ is a random element of $S$ that is labeled as 1 in $D_c$, and that an $\alpha$-consistent hypothesis for $D_c$ must label at least $(1-\frac{1}{8})d$ such elements as 1. Hence, by the utility properties of $\cA$, we have that
$$
\Pr[h(x)=1]\geq(1-\beta)(1-1/8)\geq3/4.
$$
Similarly, $x$ is a random elements of $S$ that is labeled as 0 in $D_{c'}$, and an $\alpha$-consistent hypothesis for $D_{c'}$ must not label more than $d/8$ such elements as 1. Hence,
$$
\Pr[h'(x)=1]\leq\beta+(1-\beta)\frac{1}{8}\leq1/4.
$$

Finally, as $D_c$ and $D_{c'}$ differ in at most $16\alpha n/d$ entries, differential privacy ensures that
$$
3/4\leq\Pr[h(x)=1]
\leq e^{16 \alpha \epsilon n/d}\cdot\Pr[h'(x)=1]+e^{16 \alpha \epsilon n/d}\cdot 16 \alpha n \delta/d
\leq e^{16 \alpha \epsilon n/d}\cdot1/2,
$$
showing that $n\geq\frac{d}{40\alpha\epsilon}$.
\end{proof}

\subsection{Private Learning of Thresholds vs.\ the Interior Point Problem}\label{sec:PPAC-vs-IPP}
}
\stoctext{\paragraph{Private Learning of Thresholds vs.\ the Interior Point Problem.}}

We show that with differential privacy, there is a $\Theta(1/\alpha)$ multiplicative relationship between the sample complexities of properly PAC learning thresholds with $(\alpha, \beta)$-accuracy and of solving the interior point problem with error probability $\Theta(\beta)$. Specifically, we show


\begin{theorem}\label{thm:learning-vs-range}
Let $X$ be a totally ordered domain. Then,
\begin{enumerate}
	\item If there exists an $(\eps, \delta)$-differentially private algorithm solving the interior point problem on $X$ with error probability $\beta$ and sample complexity $n$, then there is a $(2\eps, (1 + e^{\eps})\delta)$-differentially private $(2\alpha, 2\beta)$-accurate proper PAC learner for
\stoctext{threshold functions over $X$}\stocrm{$\thresh_X$}
with sample complexity $\max \left\{\frac{n}{2\alpha}, \frac{4 \log(2/\beta)}{\alpha}\right\}$.

	\item If there exists an $(\eps,\delta)$-differentially private $(\alpha, \beta)$-accurate proper PAC learner for
\stoctext{thresholds over $X$}
\stocrm{$\thresh_X$}	
with sample complexity $n$, then there is a $(2\eps, (1 + e^{\eps})\delta)$-differentially private algorithm that solves the interior point problem on $X$ with error $\beta$ and sample complexity $27\alpha n$.
\end{enumerate}

\end{theorem}

We show this equivalence in two phases. In the first, we show a $\Theta(1/\alpha)$ relationship between the sample complexity of solving the interior point problem and the sample complexity of
\stoctext{identifying an $\alpha$-consistent hypothesis for every fixed input database.}
\stocrm{empirically learning thresholds.}
We then use generalization and resampling arguments to show that with privacy, this latter task is equivalent to learning with samples from a distribution.

\stocrm{
\begin{lemma}\label{lem:range-learning-reduction}
Let $X$ be a totally ordered domain. Then,
\begin{enumerate}
	\item If there exists an $(\eps, \delta)$-differentially private algorithm solving the interior point problem on $X$ with error probability $\beta$ and sample complexity $n$, then there is a $(2\eps, (1 + e^{\eps})\delta)$-differentially private algorithm for properly and empirically learning thresholds with $(\alpha, \beta)$-accuracy and sample complexity $n / (2\alpha)$.

	\item If there exists an $(\eps,\delta)$-differentially private algorithm that is able to properly and empirically learn thresholds on $X$ with $(\alpha,\beta)$-accuracy and sample complexity $n/(3\alpha)$, then there is a $(2\eps, (1 + e^{\eps})\delta)$-differentially private algorithm that solves the interior point problem on $X$ with error $\beta$ and sample complexity $n$.
\end{enumerate}

\end{lemma}

\begin{proof}
For the first direction, let $\cA$ be a private algorithm for the interior point problem on databases of size $n$. Consider the algorithm $\cA'$ that, on input a database $D$ of size $n/(2\alpha)$, runs $\cA'$ on a database $D'$ consisting of the largest $n/2$ elements of $D$ that are labeled $1$ and the smallest $n/2$ elements of $D$ that are labeled $0$. If there are not enough of either such element, pad $D'$ with $\min\{X\}$'s or $\max\{X\}$'s respectively. Note that if $x$ is an interior point of $D'$ then $c_x$ is a threshold function with error at most $\frac{n/2}{n/(2\alpha)}$ on $D$, and is hence $\alpha$-consistent with $D$.
For privacy, note that changing one row of $D$ changes at most two rows of $D'$. Hence, applying algorithm $\cA$ preserves $(2\eps,(e^\eps+1)\delta)$-differential privacy.

For the reverse direction, suppose $\cA'$ privately finds an $\alpha$-consistent threshold functions for databases of size $n/(3\alpha)$. Define $\cA$ on a database $D'\in X^n$ to label the smaller $n/2$ points 1 and the larger $n/2$ points 0 to obtain a labeled database $D\in(X\times\{0,1\})^n$, pad $D$ with an equal number of $(\min\{X\},1)$ and $(\max\{X\},0)$ entries to make it of size $n/(3\alpha)$, and run $\cA'$ on the result. Note that if $c_x$ is a threshold function with error at most $\alpha$ on $D$ then $x$ is an interior point of $D'$, as otherwise $c_x$ has error at least $\frac{n/2}{n/(3\alpha)}>\alpha$ on $D$.
For privacy, note that changing one row of $D'$ changes at most two rows of $D$. Hence, applying algorithm $\cA'$ preserves $(2\eps,(e^\eps+1)\delta)$-differential privacy.
\end{proof}

Now we show that the task of privately outputting an almost consistent hypothesis on any fixed database is essentially equivalent to the task of private (proper) PAC learning. One direction follows immediately from a standard generalization bound for learning thresholds:

%

\begin{lemma}
Any algorithm $\cA$ for empirically learning $\thresh_X$ with $(\alpha,\beta)$-accuracy is also a $(2\alpha, \beta + \beta')$-accurate PAC learner for $\thresh_X$ when given at least $\max \{n, 4\ln (2/\beta')/\alpha\}$ samples.
\end{lemma}

\begin{proof}
Let $\cD$ be a distribution over a totally ordered domain $X$ and fix a target concept $c=q_x \in \thresh_X$. It suffices to show that for a sample $S=((x_i,c(x_i)), \dots (x_m, c(x_m)))$ where $m\geq4\ln (2/\beta')/\alpha$
and the $x_i$ are drawn i.i.d.\ from $\cD$, it holds that
\[\Pr\left[\exists \; h\in C:\;\; \error_\cD(h,c) > 2\alpha\ \land \  \error_S(h)\leq\alpha \right] \le \beta'.\]

Let $x^- \le x$ be the largest point with $\error_{\cD}(q_{x^-}, c) \ge 2\alpha$. If some $y \le x$ has $\error_{\cD}(q_{y}, c) \ge 2\alpha$ then $y \le x^{-}$, and hence for any sample $S$, $\error_S(q_{x^-}) \le \error_S(q_{y})$. Similarly let $x^+ \ge x$ be the smallest point with $\error_{\cD}(q_{x^+}, c) \ge 2\alpha$. Let $c^- = q_{x^-}$ and $c^+ = q_{x^+}$. Then it suffices to show that
\[\Pr\left[\error_S(c^-)\leq\alpha \lor \error_S(c^+) \leq\alpha\right] \le \beta'.\]
Concentrating first on $c^-$, we define the error region $R^- = (x^-, x] \cap X$ as the interval where $c^-$ disagrees with $c$. By a Chernoff bound, the probability that after $m$ independent samples from $\cD$, fewer than $\alpha m$ appear in $R^-$ is at most $\exp(-\alpha m / 4) \le \beta'/2$. The same argument holds for $c^+$, so the result follows by a union bound.
\end{proof}

In general, an algorithm that can output an $\alpha$-consistent hypothesis from concept class $\cC$ can also be used to learn $\cC$ using $\max \{n, 64\VC(\cC) \log (512/\alpha \beta')/\alpha\}$ samples \cite{BlumerEhHaWa89}. The concept class of thresholds has VC dimension $1$, so the generalization bound for thresholds saves an $O(\log(1/\alpha))$ factor over the generic statement.

For the other direction, we note that a distribution-free learner must perform well on the uniform distribution on the rows of any fixed database, and thus must be useful for outputting a consistent hypothesis on such a database.

\begin{lemma}
Suppose $\cA$ is an $(\epsilon,\delta)$-differentially private $(\alpha, \beta)$-accurate PAC learner for a concept class $\cC$ with sample complexity $m$. Then there is an $(\epsilon,\delta)$-differentially private $(\alpha, \beta)$-accurate empirical learner for $\cC$ with sample complexity $n=9m$. Moreover, if $\cA$ is proper, then so is the resulting empirical learner.
\end{lemma}

\begin{proof}
Consider a database $D = \{(x_i, y_i)\} \in (X \times \{0, 1\})^n$. Let $\cD$ denote the uniform distribution over the rows of $D$. Then drawing $m$ i.i.d.\ samples from $\cD$ is equivalent to subsampling $m$ rows of $D$ (with replacement). Consider the algorithm $\tilde{\cA}$ that subsamples (with replacement) $m$ rows from $D$ and runs $\cA$ on it. Then with probability at least $1 - \beta$, algorithm $\cA$ outputs an $\alpha$-good hypothesis on $\cD$, which is in turn an $\alpha$-consistent hypothesis for $D$. Moreover, by Lemma~\ref{lem:secrecy-of-the-sample} (secrecy-of-the-sample), algorithm $\cA$ is $(\eps,\delta)$-differentially private.
\end{proof}

}

\stocrm{
\section{Thresholds in High Dimension}

We next show that the bound of $\Omega(\log^*|X|)$ on the sample complexity of private proper-learners for $\thresh_X$ extends to conjunctions of $\ell$ independent threshold functions in $\ell$ dimensions. As we will see, every private proper-learner for this class requires a sample of $\Omega(\ell \cdot \log^*|X|)$ elements. This also yields a similar lower bound for the task of query release, as in general an algorithm for query release can be used to construct a private learner.

The significance of this lower bound is twofold. First, for reasonable settings of parameters (e.g. $\delta$ is negligible and items in $X$ are of polynomial bit length in $n$), our $\Omega(\log^*|X|)$ lower bound for threshold functions is dominated by the dependence on $\log(1/\delta)$ in the upper bound. However, $\ell \cdot \log^*|X|$ can still be much larger than $\log(1/\delta)$, even when $\delta$ is negligible in the bit length of items in $X^\ell$. Second, the lower bound for threshold functions only yields a separation between the sample complexities of private and non-private learning for a class of VC dimension $1$. Since the concept class of $\ell$-dimensional thresholds has VC dimension of $\ell$, we obtain an $\omega(\VC(C))$ lower bound for concept classes even with arbitrarily large VC dimension.

Consider the following extension of $\thresh_X$ to $\ell$ dimensions.

\begin{definition}
For a totally ordered set $X$ and $\vec{a}=(a_1,\ldots,a_{\ell})\in X^{\ell}$ define the concept $c_{\vec{a}}:X^{\ell}\rightarrow\{0,1\}$ where $c_{\vec{a}}(\vec{x})=1$ if and only if for every $1\leq i\leq \ell$ it holds that $x_i\leq a_i$. Define the concept class of all thresholds over $X^{\ell}$ as $\thresh_X^{\ell}=\{ c_{\vec{a}} \}_{\vec{a}\in X^{\ell}}$.
\end{definition}

Note that the VC dimension of $\thresh_X^{\ell}$ is $\ell$. We obtain the following lower bound on the sample complexity of privately learning $\thresh_X^{\ell}$.

\begin{theorem} \label{thm:lb-higher}
For every $n, \ell \in \N$, $\alpha > 0$, and $\delta \le \ell^2/(1500n^2)$, any $(\eps = \frac{1}{2}, \delta)$-differentially private and $(\alpha, \beta = \frac{1}{8})$-accurate proper learner for $\thresh_X^{\ell}$ requires sample complexity $n=\Omega(\frac{\ell}{\alpha}\log^*|X|)$.
\end{theorem}

This is the result of a general hardness amplification theorem for private proper learning. We show that if privately learning a concept class $C$ requires sample complexity $n$, then learning the class $C^\ell$ of conjunctions of $\ell$ different concepts from $C$ requires sample complexity $\Omega(\ell n)$.

\begin{definition}
For $\ell \in \N$, a data universe $X$ and a concept class $C$ over $X$, define a concept class $C^\ell$ over $X^\ell$ to consist of all $\vec{c} = (c_1, \dots, c_\ell)$, where $\vec{c}: X^\ell \to \{0, 1\}$ is defined by $\vec{c}(\vec{x}) = c_1(x_1) \land c_2(x_2) \land \dots \land c_{\ell}(x_{\ell})$.
\end{definition}

\begin{theorem} \label{thm:lb-higher-general}
Let $\alpha, \beta, \eps, \delta > 0$. Let $C$ be a concept class over a data universe $X$, and assume there is a domain element $p_1\in X$ s.t.\ $c(p_1)=1$ for every $c\in C$. Let $\cD$ be a distribution over databases containing $n$ examples from $X$ labeled by a concept in $C$, and suppose that every $(\eps, \delta)$-differentially private algorithm fails to find an $(\alpha/\beta)$-consistent hypothesis $h \in C$ for $D \sim \cD$ with probability at least $2\beta$. Then any $(\eps, \delta)$-differentially private and $(\alpha, \beta)$-accurate proper learner for $C^\ell$ requires sample complexity $\Omega(\ell n)$.
\end{theorem}

Note that in the the above theorem we assumed the existence of a domain element $p_1\in X$ on which every concept in $C$ evaluates to 1. To justify the necessity of such an assumption, consider the class of {\em point functions} over a domain $X$ defined as $\point_X=\{ c_x : x\in X \}$ where $c_x(y)=1$ iff $y=x$. As was shown in~\cite{BeimelNiSt13b}, this class can be privately learned using $O_{\alpha,\beta,\epsilon,\delta}(1)$ labeled examples (i.e., the sample complexity has no dependency in $|X|$). Observe that since there is no $x\in X$ on which every point concept evaluates to 1, we cannot use Theorem~\ref{thm:lb-higher-general} to lower bound the sample complexity of privately learning $\point_X^{\ell}$. Indeed, the class $\point_X^{\ell}$ is identical (up to renaming of domain elements) to the class $\point_{X\ell}$, and can be privately learned using $O_{\alpha,\beta,\epsilon,\delta}(1)$ labeled examples.

\begin{remark}
Similarly to Theorem~\ref{thm:lb-higher-general} it can be shown that if privately learning a concept class $C$ requires sample complexity $n$, and if there exists a domain element $p_0\in X$ s.t.\ $c(p_0)=0$ for every $c\in C$, then learning the class of {\em disjunctions} of $\ell$ concepts from $C$ requires sample complexity $\ell n$.
\end{remark}

\begin{proof}[Proof of Theorem~\ref{thm:lb-higher-general}]
Assume toward a contradiction that there exists an $(\eps, \delta)$-differentially private and $(\alpha, \beta)$-accurate proper learner $\cA$ for $C^{\ell}$ using $\ell n / 9$ samples.
Recall that the task of privately outputting a good hypothesis on any fixed database is essentially equivalent to the task of private PAC learning (See Section~\ref{sec:PPAC-vs-IPP}).
We can assume, therefore, that $\cA$ outputs an $\alpha$-consistent hypothesis for every fixed database of size at least $n' \triangleq \ell n$ with probability at least $1 - \beta$.

We construct an algorithm ${Solve}_{\cD}$ which uses $\cA$ in order to find an $(\alpha/\beta)$-consistent threshold function for databases of size $n$ from $\cD$. Algorithm ${Solve}_{\cD}$ takes as input a set of $n$ labeled examples in $X$ and applies $\cA$ on a database containing $n'$ labeled examples in $X^{\ell}$.
The $n$ input points are embedded along one random axis, and random samples from $\cD$ are placed on each of the other axes (with $n$ labeled points along each axis).

\begin{algorithm}[H]
\caption{${Solve}_{\cD}$}
{\bf Input:} Database $D=(x_i,y_i)_{i=1}^{n} \in (X \times \{0,1\})^{n}$.
\begin{enumerate}[rightmargin=10pt,itemsep=1pt]

\item Initiate $S$ as an empty multiset.
\item Let $r$ be a (uniform) random element from $\{1,2,\ldots,\ell\}$.
\item For $i=1$ to $n$, let $\vec{z_i}\in X^{\ell}$ be the vector with $r^\text{th}$ coordinate $x_i$, and all other coordinates $p_1$ (recall that every concept in $C$ evaluates to 1 on $p_1$). Add to $S$ the labeled example $(\vec{z_i},y_i)$.
\item For every axis $t\neq r$:
\begin{enumerate}
	\item Let $D'=(x'_i,y'_i)_{i=1}^{n} \in (X \times \{0,1\})^{n}$ denote a (fresh) sample from $\cD$.
	\item For $i=1$ to $n$, let $\vec{z'_i}\in X^{\ell}$ be the vector whose $t^\text{th}$ coordinate is $x'_i$, and its other coordinates are $p_1$. Add to $S$ the labeled example $(\vec{z'_i},y'_i)$.
\end{enumerate}
\item Let $(h_1,h_2,\ldots,h_{\ell})=\vec{h}\leftarrow\cA(S)$.
\item Return $h_r$.
\end{enumerate}
\end{algorithm}

First observe that ${Solve}_{\cD}$ is $(\eps,\delta)$-differentially private. To see this, note that a change limited to one input entry affects only one entry of the multiset $S$. Hence, applying the $(\eps,\delta)$-differentially private algorithm $\cA$ on $S$ preserves privacy.

Consider the execution of ${Solve}_{\cD}$ on a database $D$ of size $n$, sampled from $\cD$.
We first argue that $\cA$ is applied on a multiset $S$ correctly labeled by a concept from $C^{\ell}$.
For $1\leq t\leq\ell$ let $(x^t_i,y^t_i)_{i=1}^{n}$ be the sample from $\cD$ generated for the axis $t$, let $(\vec{z^t_i},y^t_i)_{i=1}^{n}$ denote the corresponding elements that were added to $S$, and let $c_t$ be s.t.\ $c_t(x^t_i)=y^t_i$ for every $1\leq i \leq n$. Now observe that
$$
(c_1,c_2,\ldots,c_{\ell})(\vec{z^t_i})=c_1(p_1)\land c_2(p_1)\land \dots \land c_t(x^t_i)\land\dots\land c_{\ell}(p_1) = y^t_i,
$$
and hence $S$ is perfectly labeled by $(c_1,c_2,\ldots,c_{\ell})\in C^{\ell}$.

By the properties of $\cA$, with probability at least $1-\beta$ we have that $\vec{h}$ (from Step~5) is an $\alpha$-consistent hypothesis for $S$.
Assuming that this is the case, there could be at most $\beta\ell$ ``bad'' axes on which $\vec{h}$ errs on more than $\alpha n / \beta$ points.
Moreover, as $r$ is a random axis, and as the points along the $r^\text{th}$ axis are distributed exactly like the points along the other axes, the probability that $r$ is a ``bad'' axis is at most $\frac{\beta\ell}{\ell}=\beta$.
Overall, ${Solve}_{\cD}$ outputs an $(\alpha/\beta)$-consistent hypothesis with probability at least $(1-\beta)^2>1-2\beta$. This contradicts the hardness of the distribution $\cD$.
\end{proof}

Now the proof of Theorem \ref{thm:lb-higher} follows from the lower bound on the sample complexity of privately finding an $\alpha$-consistent threshold function (see Section~\ref{sec:lowerBound}):

\begin{lemma}[Follows from Lemma~\ref{lem:range-lb-dist} and~\ref{lem:range-learning-reduction}]\label{lemma:lowerRestated}
There exists a constant $\lambda>0$ s.t.\ the following holds.
For every totally ordered data universe $X$ there exists a distribution $\cD$ over databases containing at most $n=\frac{\lambda}{\alpha}\log^*|X|$ labeled examples from $X$ such that every $(\frac{1}{2}, \frac{1}{50 n^2})$-differentially private algorithm fails to find an $\alpha$-consistent threshold function for $D\sim\cD$ with probability at least $\frac{1}{4}$.
\end{lemma}

We remark that, in general, an algorithm for query release can be used to construct a private learner with similar sample complexity. Hence, Theorem~\ref{thm:lb-higher} also yields the following lower bound on the sample complexity of releasing approximated answers to queries from $\thresh_X^{\ell}$.

\begin{theorem}\label{thm:query-release-lower-bound}
For every $n, \ell \in \N$, $\alpha > 0$, and $\delta \le \ell^2/(7500n^2)$, any $(\frac{1}{150}, \delta)$-differentially private algorithm for releasing approximated answers for queries from $\thresh_X^{\ell}$ with $(\alpha,\frac{1}{150})$-accuracy must have sample complexity $n=\Omega(\frac{\ell}{\alpha}\log^*|X|)$.
\end{theorem}

In order to prove the above theorem we use our lower bound on privately learning $\thresh_X^{\ell}$ together with the following reduction from private learning to query release.

\begin{lemma}[\cite{GuptaHaRoUl11,BeimelNiSt13b}]\label{lem:sanitization-implies-PPAC}
Let $C$ be a class of predicates.
If there exists a $(\frac{1}{150},\delta)$-differentially private algorithm capable of releasing queries from $C$ with $(\frac{1}{150},\frac{1}{150})$-accuracy and sample complexity $n$, then there exists a $(\frac{1}{5},5\delta)$-differentially private $(\frac{1}{5},\frac{1}{5})$-accurate PAC learner for $C$ with sample complexity $O(n)$.
\end{lemma}

\begin{proof}[Proof of Theorem~\ref{thm:query-release-lower-bound}]
Let $\delta\leq\ell^2/(7500n^2)$.
Combining our lower bound on the sample complexity of privately learning $\thresh_X^{\ell}$ (Theorem~\ref{thm:lb-higher}) together with the reduction stated in Lemma~\ref{lem:sanitization-implies-PPAC}, we get a lower bound of $m\triangleq\Omega(\ell\cdot\log^*|X|)$ on the sample complexity of every $(\frac{1}{150},\delta)$-differentially private algorithm for releasing queries from $\thresh_X^{\ell}$ with $(\frac{1}{150},\frac{1}{150})$-accuracy.

In order to refine this argument and get a bound that incorporates the approximation parameter, let $\alpha\leq1/150$, and assume towards contradiction that there exists a $(\frac{1}{150},\delta)$-differentially private algorithm $\tilde{\cA}$ for releasing queries from $\thresh_X^{\ell}$ with $(\alpha,\frac{1}{150})$-accuracy and sample complexity $n < m/(150\alpha)$.

We will derive a contradiction by using $\tilde{\cA}$ in order to construct a $(\frac{1}{150},\frac{1}{150})$-accurate algorithm for releasing queries from $\thresh_X^{\ell}$ with sample complexity less than $m$.
Consider the algorithm $\cA$ that on input a database $D$ of size $150\alpha n$, applies $\tilde{\cA}$ on a database $\tilde{D}$ containing the elements in $D$ together with $(1-150\alpha)n$ copies of $(\min X)$. Afterwards, algorithm $\cA$ answers every query $c\in\thresh_X^{\ell}$ with $a_c \triangleq \frac{1}{150\alpha}(\tilde{a_c} - 1 + 150\alpha)$, where $\{\tilde{a_c}\}$ are the answers received from $\tilde{\cA}$.

Note that as $\tilde{\cA}$ is $(\frac{1}{150},\delta)$-differentially private, so is $\cA$.
We now show that $\cA$'s output is $\frac{1}{150}$-accurate for $D$ whenever $\tilde{\cA}$'s output is $\alpha$-accurate for $\tilde{D}$, which happens with all but probability $\frac{1}{150}$.
Fix a query $c\in\thresh_X^{\ell}$ and assume that $c(D)=t/(150\alpha n)$. Note that $c(\min X)=1$, and hence, $c(\tilde{D})=t/n + (1 - 150\alpha)$. By the utility properties of $\tilde{\cA}$,
\begin{eqnarray*}
a_c &=& \frac{1}{150\alpha}(\tilde{a_c} - 1 + 150\alpha) \\
&\leq& \frac{1}{150\alpha}(c(\tilde{D}) + \alpha - 1 + 150\alpha) \\
&=& \frac{1}{150\alpha}(t/n + \alpha)\\
&=& t/(150\alpha n) + 1/150\\
&=& c(D)+1/150.
\end{eqnarray*}
Similar arguments show that $a_c\geq c(D) - 1/150$, proving that $\cA$ is $(1/150,1/150)$-accurate and contradicting the lower bound on the sample complexity of such algorithms.
\end{proof}

}

\stocrm{

\section{Mechanism-Dependent Lower Bounds}

\subsection{The Undominated Point Problem}

By a reduction to the interior point problem problem, we can prove an impossibility result for the problem of privately outputting something that is at least the minimum of a database on an unbounded domain. Specifically, we show

\begin{theorem} \label{thm:half-range}
For every (infinite) totally ordered domain $X$ with no maximum element (e.g., $X=\N$) and every $n\in\N$, there is no $(\eps, \delta)$-differentially private mechanism
$M:X^n\rightarrow X$ such that for every $x=(x_1,\ldots,x_n)\in X^n$,
$$\Pr[M(x)\geq\min_i x_i]\geq2/3.$$
\end{theorem}

Besides being a natural relaxation of the interior point problem, this \emph{undominated point problem} is of interest because we require new techniques to obtain lower bounds against it. Note that if we ask for a mechanism that works over a bounded domain (e.g., $[0,1]$), then the problem is trivial. Moreover, this means that proving a lower bound on the problem when the domain is $\N$ cannot possibly go by way of constructing a single distribution that every differentially private mechanism fails on. The reason is that for any distribution $\cD$ over $\N^n$, there is some number $K$ where $\Pr_{D \getsr \cD}[\max{D} > K] \le 2/3$, so the mechanism that always outputs $K$ solves the problem.

\begin{proof}
Without loss of generality we may take $X=\N$, since every totally ordered domain with no maximum element contains an infinite sequence $x_0<x_1<x_2<x_3<\dots$.
To prove our lower bound we need to take advantage of the fact that we only need to show that for each differentially private mechanism $M$ there exists a distribution, depending on $M$, over which $M$ fails. To this end, for an increasing function $T:\N \to \N$, we say that a mechanism $M:\N^n \to \N$ is ``$T$-bounded'' if $\Pr[M(x_1,\dots,x_n) \ge T(\max_i x_i)] < 1/8$. That is, $M$ is $T$-bounded if it is unlikely to output anything larger than $T$ applied to the max of its input. Note that any mechanism is $T$-bounded for some function $T$.

We can then reduce the impossibility of the undominated point problem for $T$-bounded mechanisms to our lower bound for the interior point problem. First, fix a function $T$. Suppose for the sake of contradiction that there were a $T$-bounded mechanism $M$ that solves the undominated point problem on $(x_1,\dots,x_n)$ with probability at least $7/8$. Then by a union bound, $M$ must output something in the interval $[\min_i x_i, T(\max_i x_i))$ with probability at least $3/4$. Now, for $d\in\N$, consider the data universe $X_d=\{ 1, T(1), T(T(1)), T(T(T(1))), \dots, T^{(d-1)}(1)  \}$ and the differentially private mechanism $M':X_d^n\rightarrow X_d$ that, on input a database $D$ runs $M(D)$ and rounds the answer down to the nearest $T^i(d)$. Then $M'$  solves the interior point problem on the domain $X_d$ with probability at least $3/4$. By our lower bound for the interior point problem we have $n = \Omega(\log^* d)$, which is a contradiction since $n$ is fixed and $d$ is arbitrary.

\end{proof}
}

\stocrm{
\subsection{Properly Learning Point Functions with Pure Differential Privacy}

Using similar ideas as in the above section, we revisit the problem of privately learning the concept class $\point_{\N}$ of point functions over the natural numbers. Recall that a point function $c_x$ is defined by $c_x(y) = 1$ if $x = y$ and evaluates to $0$ otherwise. Beimel et al. \cite{BeimelKaNi10} used a packing argument to show that $\point_{\N}$ cannot be properly learned with {\em pure} $\eps$-differential privacy (i.e., $\delta{=}0$). However, more recent work of Beimel et al. \cite{BeimelNiSt13a} exhibited an $\eps$-differentially private \emph{improper} learner for $\point_{\N}$ with sample complexity $O(1)$. Their construction required an uncountable hypothesis class, with each concept being described by a real number. This left open the question of whether $\point_{\N}$ could be learned with a countable hypothesis class, with each concept having a finite description length.

We resolve this question in the negative. Specifically, we show that it is impossible to learn (even improperly) point functions over an infinite domain with pure differential privacy using a countable hypothesis class.

%

\begin{theorem} \label{thm:points}
Let $X$ be an infinite domain, let $H$ be a countable collection of hypotheses $\{h : X \to \{0, 1\}\}$, and let $\eps\geq 0$. Then there is no $\eps$-differentially private $(1/3, 1/3)$-accurate PAC learner for points over $X$ using the hypothesis class $H$.
\end{theorem}

\begin{remark}
A learner implemented by an algorithm (i.e. a probabilistic Turing machine) must use a hypothesis class where each hypothesis has a finite description. Note that the standard proper learner for $\point_X$ can be implemented by an algorithm. However, a consequence of our result is that there is no algorithm for privately learning $\point_X$.
\end{remark}

\begin{proof}
For clarity, and without loss of generality, we assume that $X=\N$.
Suppose for the sake of contradiction that we had an $\eps$-differentially private learner $M$ for point functions over $\N$ using hypothesis class $H$. Since $H$ is countable, there is a finite subset of hypotheses $H'$ such that $M((0, 1)^n) \in H'$ with probability at least $5/6$, where $(0, 1)^n$ is the dataset where all examples are the point $0$ with the label $1$. Indeed
$\sum_{h\in H}\Pr[M((0, 1)^n)=h]=1$, so some finite partial sum of this series is at least $5/6$.
Now to each point $x \in \N$ we will associate a distribution $\cD_x$ on $\N$ and let $G_x \subseteq H'$ be the set of hypotheses $h$ \emph{in the finite set $H'$} for which
\[\Pr_{y \sim \cD_x}[c_x(y) = h(y)] \ge 2/3.\]
We establish the following claim.
\begin{claim} \label{claim:disjoint-good-hyp}
There is an infinite sequence of points $x_1,x_2,x_3,\dots$ together with distributions $\cD_i := \cD_{x_i}$ such that the sets $G_i := G_{x_i}$ are all disjoint.
\end{claim}
Given the claim, the result follows by a packing argument \cite{HardtTa10, BeimelKaNi10}. By the utility of $M$, for each $\cD_i$ there is a database $R_i \in (\N \times \{0, 1\})^n$ in the support of $\cD_i^n$ such that $\Pr[M(R_i) \in G_i] \ge 2/3 - 1/6 = 1/2$. By changing the database $R_i$ to $(0, 1)^n$ one row at a time while applying the differential privacy constraint, we see that
\[\Pr[M((0, 1)^n) \in G_i] \ge \frac{1}{2}e^{-\eps n}.\]
It is impossible for this to hold for infinitely many disjoint sets $G_i$.
\end{proof}

\begin{proof}[Proof of Claim \ref{claim:disjoint-good-hyp}]
We inductively construct the sequence $(x_i)$, starting with $x_1 = 0$. Now suppose we have constructed $x_1, \dots, x_i$ with corresponding good hypothesis sets $G_1, \dots, G_i$. Let $B = \cup_{j = 1}^i G_i$ be the set of hypotheses with wish to avoid. Note that $B$ is a finite set of hypotheses, so there are some $x, x' \in \N$ for which every $h \in B$ with $h(x) = 1$ also has $h(x') = 1$. Let $x_{i+1} = x$ and $\cD_{i}$ be distributed uniformly over $x$ and $x'$. Then for every hypothesis $h \in B$,
\[\Pr_{y \sim D_i}[c_{x_{i+1}}(y) = h(y)] \le 1/2,\]
and hence $G_{i+1}$ is disjoint from the preceding $G_j$'s.
\end{proof}
}

\paragraph{Acknowledgments.} We thank Amos Beimel, Adam Smith, \stoctext{and} Jonathan Ullman\stocrm{, and anonymous reviewers} for helpful conversations and suggestions that helped guide our work. We thank Gautam Kamath for pointing us to references on distribution learning. We also thank Haim Kaplan and Roodabeh Safavi for pointing to us an error in the proof of Theorem~\ref{thm:sanitization-vs-range} in an earlier version of this paper.

\bibliographystyle{alpha}

\newcommand{\etalchar}[1]{$^{#1}$}

\stocrm{

\appendix

\section{The Choosing Mechanism} \label{app:choosing}

We supply the proofs of privacy and utility for the choosing mechanism.

\begin{proof}[Proof of Lemma \ref{lem:CMprivacy}]
Let $A$ denote the choosing mechanism (Algorithm \ref{alg:choosing}). Let $S,S'$ be neighboring databases of $m$ elements. We need to show that $\Pr[A(S)\in R]\leq \exp(\epsilon)\cdot\Pr[A(S')\in R]+\delta$ for every set of outputs $R \subseteq \cF \cup \{\bot\}$. Note first that $\OPT(S) = \max_{f \in \cF} \{q(S, f)\}$ has sensitivity at most $1$, so by the properties of the Laplace Mechanism,
\begin{eqnarray}\label{eq:CMbot}
\Pr[A(S)=\bot] &=& \Pr\left[\widetilde{\OPT}(S)<\frac{8}{\epsilon}\ln(\frac{4k}{\beta\epsilon\delta})\right] \nonumber\\
&\leq& \exp(\frac{\epsilon}{4})\cdot\Pr\left[\widetilde{\OPT}(S')<\frac{8}{\epsilon}\ln(\frac{4k}{\beta\epsilon\delta})\right] \nonumber\\
&=& \exp(\frac{\epsilon}{4})\cdot\Pr[A(S')=\bot].
\end{eqnarray}
Similarly, we have $\Pr[A(S) \ne \bot] \le \exp(\eps/4) \Pr[A(S') \ne \bot]$. Thus, we my assume below that $\bot \not\in R$. (If $\bot \in R$, then we can write $\Pr[A(S) \in R] = \Pr[A(S) = \bot] + \Pr[A(S) \in R \setminus \{\bot\}]$, and similarly for $S'$.)

\paragraph{Case~(a): $\OPT(S) < \frac{4}{\epsilon}\ln(\frac{4k}{\beta\epsilon\delta})$.} It holds that
\begin{align*}
\Pr[A(S) \in R] &\le \Pr[A(S)\neq\bot] \\
&\leq \Pr\left[\Lap\left(\frac{4}{\epsilon}\right)>\frac{4}{\epsilon}\ln\left(\frac{4k}{\beta\epsilon\delta}\right)\right] \\
&\le \delta \le \Pr[A(S') \in R] + \delta.
\end{align*}

\paragraph{Case~(b): $\OPT(S)\geq\frac{4}{\epsilon}\ln(\frac{4k}{\beta\epsilon\delta})$.} Let $G(S)$ and $G(S')$ be the sets used in step~2 in the execution $S$ and on $S'$ respectively.
We will show that the following two facts hold:\\

\noindent
$Fact \; 1:$ For every $f\in G(S)\setminus G(S')$, it holds that $\Pr[A(S)=f]\leq \frac{\delta}{k}$.\\

\noindent
$Fact \; 2:$ For every possible output $ f \in G(S) \cap G(S')$, it holds that $\Pr[A(S)=f]\leq e^{\epsilon} \cdot\Pr[A(S')=f]$.\\

We first show that the two facts imply that the lemma holds for Case~(b).
Let $B \triangleq G(S)\setminus G(S')$, and note that as $q$ is of $k$-bounded growth, $|B|\leq k$.
Using the above two facts, for every set of outputs $R \subseteq \cF$ we have
\begin{eqnarray*}
\Pr[A(S)\in R] &=&  \Pr[A(S)\in R\setminus B] + \sum_{f \in R \cap B}\Pr[A(S)=f]\\
&\leq& e^\epsilon \cdot \Pr[A(S')\in R\setminus B] + |R \cap B|\frac{\delta}{k}\\
&\leq& e^\epsilon \cdot \Pr[A(S')\in R] + \delta.\\
\end{eqnarray*}

To prove Fact~1, let $f\in G(S)\setminus G(S')$.
That is, $q(S,f)\geq1$ and $q(S',f)=0$. As $q$ has sensitivity at most $1$, it must be that $q(S,f)=1$. As there exists $\hat f\in S$ with $q(S,\hat f)\geq\frac{4}{\epsilon}\ln(\frac{4k}{\beta\epsilon\delta})$, we have that
$$
\Pr[A(S)=f]\leq
\Pr\left[  \begin{array}{c}
	\text{The exponential}\\
	\text{mechanism chooses $f$}
\end{array} \right]
\leq \frac{\exp(\frac{\epsilon}{4}\cdot 1)}{\exp(\frac{\epsilon}{4}\cdot\frac{4}{\epsilon}\ln(\frac{4k}{\beta\epsilon\delta}))}
= \exp\left(\frac{\epsilon}{4}\right)\frac{\beta\eps\delta}{4k},
$$
which is at most $\delta/k$ for $\epsilon\leq2$.

To prove Fact~2, let $f \in G(S) \cap G(S')$ be a possible output of $A(S)$. We use the following Fact~3, proved below.\\

\noindent
$Fact \; 3:$ $\sum\limits_{h\in G(S')}\exp(\frac{\epsilon}{4}q(S',h))\leq e^{\epsilon/2}\cdot\sum\limits_{h\in G(S)}\exp(\frac{\epsilon}{4}q(S,h))$.\\

\noindent
Using Fact~3, for every possible output $f\in G(S)\cap G(S')$ we have that
\begin{eqnarray*}
&&\frac{\Pr[A(S)=f]}{\Pr[A(S')=f]}\\
&& \;\;\;\;\;\; =  \;\;\;
\left.\left(
\Pr[A(S)\neq\bot] \cdot
\frac{\exp(\frac{\epsilon}{4}q(f,S))}
{\sum_{h\in G(S)}\exp( \frac{\epsilon}{4} q(h,S) )}\right) \middle/ \left(
\Pr[A(S')\neq\bot] \cdot
\frac{\exp(\frac{\epsilon}{4}q(f,S'))}
{\sum_{h\in G(S')}\exp( \frac{\epsilon}{4} q(h,S') )}\right) \right.\\
&& \;\;\;\;\;\; = \;\;\;
\frac{\Pr[A(S)\neq\bot]}{\Pr[A(S')\neq\bot]}\cdot
\frac{\exp(\frac{\epsilon}{4}q(f,S))}
{\exp(\frac{\epsilon}{4}q(f,S'))} \cdot
\frac{\sum_{h\in G(S')}{\exp( \frac{\epsilon}{4} q(h,S') )}}{\sum_{h\in G(S)}{\exp( \frac{\epsilon}{4} q(h,S) )}}
\leq
e^\frac{\epsilon}{4}\cdot e^\frac{\epsilon}{4}\cdot e^\frac{\epsilon}{2}=e^\epsilon.
\end{eqnarray*}

We now prove Fact~3. Let $\XXX \triangleq \sum_{h \in G(S)} \exp(\frac{\eps}{4} q(S, h))$. Since there exists a solution $\hat{f}$ s.t. $q(S,\hat{f})\geq\frac{4}{\epsilon}\ln(\frac{4k}{\beta\epsilon\delta})$, we have
$\XXX\geq\exp(\frac{\epsilon}{4}\cdot\frac{4}{\epsilon}\ln(\frac{4k}{\beta\epsilon\delta}))\geq\frac{4k}{\epsilon}$.

Now, recall that $q$ is of $k$-bounded growth, so $|G(S')\setminus G(S)|\leq k$, and every $h\in(G(S')\setminus G(S))$ satisfies $q(S',h)=1$. Hence,
\begin{eqnarray*}
\sum_{h\in G(S')}\exp\left(\frac{\epsilon}{4}q(S',h)\right) &\leq&  k\cdot\exp\left(\frac{\epsilon}{4}\right)+ \sum_{h\in G(S')\cap G(S)}\exp\left(\frac{\epsilon}{4}q(S',h)\right)\\
&\leq&  k\cdot\exp\left(\frac{\epsilon}{4}\right)+ \exp\left(\frac{\epsilon}{4}\right)\cdot\sum_{h\in G(S')\cap G(S)}\exp\left(\frac{\epsilon}{4}q(S,h)\right)\\
&\leq&  k\cdot\exp\left(\frac{\epsilon}{4}\right)+ \exp\left(\frac{\epsilon}{4}\right)\cdot\sum_{h\in G(S)}\exp\left(\frac{\epsilon}{4}q(S,h)\right)\\
&=&  k\cdot e^{\epsilon/4}+e^{\epsilon/4}\cdot\XXX\leq e^{\epsilon/2}\XXX,\\
\end{eqnarray*}
where the last inequality follows from the fact that $\XXX \ge 4k/\eps$. This concludes the proof of Fact~3, and completes the proof of the lemma.
\end{proof}

The utility analysis for the choosing mechanism is rather straightforward:

\begin{proof}[Proof of Lemma \ref{lem:CMscoreOne}]
Recall that the mechanism defines $\widetilde{\OPT}(S)$ as $\OPT(S)+\Lap(\frac{4}{\epsilon})$.
Note that the mechanism succeeds whenever $\widetilde{\OPT}(S)\geq\frac{8}{\epsilon}\ln(\frac{4k}{\beta\epsilon\delta})$.
This happens provided the $\Lap\left(\frac{4}{\eps}\right)$ random variable is at most $\frac{8}{\eps}\ln(\frac{4k}{\beta\eps\delta})$, which happens with probability at least $(1-\beta)$.
\end{proof}

\begin{proof}[Proof of Lemma \ref{lem:CMutility}]
Note that if $\OPT(S)<\frac{16}{\epsilon}\ln(\frac{4km}{\beta\epsilon\delta})$, then every solution is a good output, and the mechanism cannot fail. Assume, therefore, that there exists a solution $f$ s.t. $q(f,S)\geq\frac{16}{\epsilon}\ln(\frac{4km}{\beta\epsilon\delta})$, and
recall that the mechanism defines $\widetilde{\OPT}(S)$ as $\OPT(S)+\Lap(\frac{4}{\epsilon})$. As in the proof of Lemma \ref{lem:CMscoreOne}, with probability at least $1 - \beta/2$, we have $\widetilde{\OPT}(S) \ge \frac{8}{\eps}\ln\left(\frac{4k}{\beta\eps\delta}\right)$. Assuming this event occurs, we will show that with probability at least $1 - \beta/2$, the exponential mechanism chooses a solution $f$ s.t. $q(S,f)\geq \opt(S)-\frac{16}{\epsilon}\ln(\frac{4km}{\beta\epsilon\delta})$.

By the growth-boundedness of $q$, and as $S$ is of size $m$, there are at most $km$ possible solutions $f$ with $q(S,f)>0$. That is, $|G(S)|\leq km$. By the properties of the Exponential Mechanism, we obtain a solution as desired with probability at least
\[\left(1- km\cdot\exp\left(-\frac{\epsilon}{4}\cdot\frac{16}{\epsilon}\ln\left(\frac{4km}{\beta\epsilon\delta}\right)\right)\right)\geq\left(1-\frac{\beta}{2}\right).\]

By a union bound, we get that the choosing mechanism outputs a good solution with probability at least $(1-\beta)$.
\end{proof}

\section{Interior Point Fingerprinting Codes} \label{app:range-fpc}

Fingerprinting codes were introduced by Boneh and Shaw \cite{BonehSh98} to address the problem of watermarking digital content. Suppose a content distributor wishes to distribute a piece of digital content to $n$ legitimate users in such a way that any pirated copy of that content can be traced back to any user who helped in producing the copy. A {\em fingerprinting code} is a scheme for assigning each $n$ users a codeword that can be hidden in their copy of the content, and then be uniquely traced back to the identity of that user. Informally, a finger printing code is {\em fully collusion-resistant} if when an arbitrary coalition $T$ of users combine their codewords to produce a new pirate codeword the pirate codeword can still be successfully traced to a member of $T$, provided the pirate codeword satisfies a certain {\em marking assumption}. Traditionally, this marking assumption requires that if all users in $T$ see the same bit $b$ at index $j$ of their codewords, then index $j$ of their combined codeword must also be $b$.

Recent work has shown how to use fingerprinting codes to obtain lower bounds in differential privacy \cite{BUV14, DworkTaThZh14, BassilySmTh14}. Roughly speaking, these works show how any algorithm with nontrivial accuracy for a given task can be used to create a pirate algorithm that satisfies the marking assumption for a fingerprinting code. The security of the fingerprinting code means that the output of this algorithm can be traced back to one of its inputs. This implies that the algorithm is not differentially private.

We show how our lower bound for privately solving the interior point problem can also be proved by the construction of an object we call an {\em interior point fingerprinting code}. The difference between this object and a traditional fingerprinting code lies in the marking assumption. Thinking of our codewords as being from an ordered domain $X$, our marking assumption is that the codeword produced by a set of $T$ users must be an interior point of their codewords. The full definition of the code is as follows.

\begin{definition} \label{def:rfpc}
For a totally ordered domain $X$, an \emph{interior point fingerprinting code} over $X$ consists of a pair of randomized algorithms $(\Gen, \Trace)$ with the following syntax.
\begin{itemize}
\item $\Gen_n$ samples a codebook $C = (x_1, \dots, x_n) \in X^n$
\item $\Trace_n(x)$ takes as input a ``codeword'' $x \in X$ and outputs either a user $i \in [n]$ or a failure symbol $\bot$.
\end{itemize}
The algorithms $\Gen$ and $\Trace$ are allowed to share a common state (e.g. their random coin tosses).

The adversary to a fingerprinting code consists of a subset $T \subseteq [n]$ of users and a pirate algorithm $\cA : X^{|T|} \to X$. The algorithm $\cA$ is given $C|_T$, i.e. the codewords $x_i$ for $i \in T$, and its output $x \getsr \cA(C|_T)$ is said to be ``feasible'' if $x \in [\min_{i\in T} x_i, \max_{i \in T} x_i]$. The security guarantee of a fingerprinting code is that for all coalitions $T \subseteq [n]$ and all pirate algorithms $\cA$, if $x = \cA(C|_T)$, then we have
\begin{enumerate}
\item Completeness: $\Pr[\Trace(x) = \bot \land x \text{ feasible}] \le \gamma$, where $\gamma
\in[0,1]$ is the {\em completeness error}.
\item Soundness: $\Pr[\Trace(x) \in [n] \setminus T] \le \xi$, where $\xi\in[0,1]$ is the {\em soundness error}.
\end{enumerate}
The probabilities in both cases are taken over the coins of $\Gen, \Trace$, and $\cA$.
\end{definition}

\begin{remark}
We note that an interior point fingerprinting code could also be interpreted as an ordinary fingerprinting code (using the traditional marking assumption) with codewords of length $|X|$ of the form $000011111$.
As an example for using such a code, consider a vendor interested in fingerprinting movies.
Using an interior point fingerprinting code, the vendor could produce fingerprinted copies by simply splicing two versions of the movie.
\end{remark}

We now argue as in \cite{BUV14} that the existence of an interior point fingerprinting code yields a lower bound for privately solving the interior point problem.

\begin{lemma} \label{lem:rfpc-to-dp}
Let $\eps\le1$, $\delta\le 1/(12n)$, $\gamma\le1/2$ and $\xi\le 1/(33n)$.
If there is an interior point fingerprinting code on domain $X$ for $n$ users with completeness error $\gamma$ and soundness error $\xi$, then there is no $(\eps, \delta)$-differentially private algorithm that, with probability at least $2/3$, solves the interior point problem on $X$ for databases of size $n-1$.
\end{lemma}

\begin{proof}
Suppose for the sake of contradiction that there were a differentially private $\cA$ for solving the interior point problem on $X^{n-1}$. Let $T = [n-1]$, and let $x = \cA(C|_T)$ for a codebook $C \getsr \Gen$.
\[1 - \gamma \leq \Pr[\Trace(x) \ne \bot \lor x \text{ not feasible}] \leq \Pr[\Trace(x) \ne \bot] + \frac{1}{3}.\]
Therefore, there exists some $i^* \in [n]$ such that
\[\Pr[\Trace(x) = i^*] \geq \frac{1}{n}\cdot\left(\frac{2}{3} - \gamma\right)\geq\frac{1}{6n}.\]
Now consider the coalition $T'$ obtained by replacing user $i^*$ with user $n$. Let $x' = \cA(C|_{T'})$, again for a random codebook $C \getsr \Gen$. Since $\cA$ is differentially private,
\[\Pr[\Trace(x') = i^*] \ge e^{-\eps}\cdot(\Pr[\Trace(x) = i^*]  - \delta) > \frac{1}{33n} \geq \xi,\]
contradicting the soundness of the interior point fingerprinting code.
\end{proof}

We now show how to construct an interior point fingerprinting code, using similar ideas as in the proof of Lemma \ref{lem:range-lb-dist}. For $n$ users, the codewords lie in a domain with size an exponential tower in $n$, allowing us to recover the $\log^*|X|$ lower bound for interior point queries.

\begin{lemma} \label{lem:rfpc-construction}
For every $n\in\N$ and $\xi>0$ there is an interior point fingerprinting code for $n$ users with completeness $\gamma = 0$ and soundness $\xi$ on a domain $X_n$ of size $|X_n| \le {\rm tower}^{(n + \log^*(2n^2/\xi))}(1)$.
\end{lemma}

\begin{proof}
Let $b(n) = 2n^2/\xi$, and define the function $S$ recursively by $S(1)=1$ and $S(n+1) = b(n)^{S(n)}$.
By induction on $n$, we will construct codes for $n$ users over a domain of size $S(n)$ with perfect completeness and soundness at most $\sum_{j = 1}^n \frac{1}{b(j)}<\xi$. First note that there is a code with perfect completeness and perfect soundness for $n = 1$ user over a domain of size $S(1)=1$. Suppose we have defined the behavior of $(\Gen_n, \Trace_n)$ for $n$ users. Then we define
\begin{itemize}
\item $\Gen_{n+1}$ samples $C' = (x_1', \dots, x_n') \getsr \Gen_n$ and $x_{n+1} \getsr [S(n+1)]$.  For each $i = 1, \dots, n$, let $x_i$ be a base-$b(n)$ number (written $x_i^{(0)} x_i^{(1)}\dots x_i^{(S(n)-1)}$, where $x_i^{(0)}$ is the most significant digit) that agrees with $x_{n+1}$ in the $x_i'$ most-significant digits, and has random entries from $[b(n)]$ at every index thereafter. The output codebook is $C = (x_1, \dots, x_{n+1})$.
\item $\Trace_{n+1}(x)$ retrieves the codebook $C$ from its shared state with $\Gen_{n+1}$. Let $M$ be the maximum number of digits to which any $x_i$ (for $i = 1, \dots, n$) agrees with $x_{n+1}$. If $x$ agrees with $x_{n+1}$ on more than $M$ digits, accuse user $n+1$. Otherwise, let $x'$ be the number of indices on which $x$ agrees with $x_{n+1}$, and run $\Trace_n(x')$ with respect to codebook $C' = (x_1', \dots, x_n')$.
\end{itemize}
We reduce the security of this scheme to that of $(\Gen_n, \Trace_n)$.
To check completeness, let $T \subseteq [n+1]$ be a pirate coalition and let $\cA$ be a pirate algorithm. Consider the pirate algorithm $\cA'$ for codes on $n$ users that, given a set of codewords $C'|_{T'}$ where $T' = T \setminus \{n+1\}$, simulates $\Gen_{n+1}$ to produce a set of codewords $C|_T$ and outputs the number $x'$ of indices on which $x=\cA(C|_T)$ agrees with $x_{n+1}$.

If $x$ is feasible for $C|_T$ and $x^{M+1} \ne x_{n+1}^{M+1}$, then $x'$ is feasible for $C'|_{T'}$. Therefore,
\begin{align*}
\Pr[\Trace_{n+1}(x) = \bot \land x \text{ feasible for  } C|_T] &= \Pr[x^{M+1} \ne x_{n+1}^{M+1} \land \Trace_{n}(x') = \bot \land x \text{ feasible for } C|_T] \\
&\le \Pr[\Trace_{n}(x') = \bot \land x' \text{ feasible for } C'|_{T'}]=0,
\end{align*}
by induction, proving perfect completeness.

To prove soundness, let $M' = \max x_i'$. Then
\begin{align*}
\Pr[\Trace_{n+1}(x) \in [n+ 1] \setminus T] &\le \Pr[\Trace_{n+1}(x) = n+1 \land (n+1)\notin T] + \Pr[\Trace_{n+1}(x) \in [n] \setminus T]\\
&\le \Pr[x^{M'+1} = x_{n+1}^{M'+1} \land (n+1)\notin T] + \Pr[\Trace_n(x') \in [n] \setminus T] \\
&\le \frac{1}{b(n)} + \sum_{j = 1}^{n-1} \frac{1}{b(j)}=\sum_{j = 1}^n \frac{1}{b(j)}<\xi.
\end{align*}
\end{proof}

Combining Lemmas \ref{lem:rfpc-to-dp} and \ref{lem:rfpc-construction} yields Theorem \ref{thm:range-lb-informal}.


}
\end{document}